\documentclass[runningheads]{llncs}
\usepackage{complexity} 
\usepackage{todonotes} 
\presetkeys{todonotes}{inline}{}
\usepackage{multirow}

\newcommand\blfootnote[1]{%
  \begingroup
  \renewcommand\thefootnote{}\footnote{#1}%
  \addtocounter{footnote}{-1}%
  \endgroup
}

\newtheorem{claim1}{Claim}[section]

\newtheorem{reduction-rule}{Reduction Rule}[section]

\newcommand{\maxedgecoloring}{\textsc{Maximum Edge Colorable Subgraph}}

\newcommand{\ilp}{\textsc{ILP}}
\newcommand{\type}{{type}}

\usepackage{amssymb,amsmath,mathrsfs}
\usepackage{soul}
\usepackage{mathtools}
\usepackage{framed}

\usepackage{listings}
\usepackage{mathrsfs}
\usepackage[mathscr]{euscript}
\usepackage{amstext}
\usepackage{graphicx,tikz}
\RequirePackage{fancyhdr}
\usepackage{boxedminipage}
\usepackage{todonotes}

\usepackage[boxed,algosection,noline]{algorithm2e}

\usepackage{graphics}
\usepackage{framed}
\usepackage{enumerate}
\usepackage{float}
\usepackage{xspace}

\newcommand{\redbluedsfull}{\textsc{Red Blue Dominating Set}}
\newcommand{\redblueds}{\textsc{RBDS}}
\newcommand{\dset}{\textsc{Dominating Set}}

\newcommand{\calA}{\mathcal{A}}

\newcommand{\calB}{\mathcal{B}}
\newcommand{\calC}{\mathcal{C}}

\newcommand{\calO}{\ensuremath{{\mathcal O}}}
\newcommand{\OO}{\mathcal{O}}
\newcommand{\Oh}{\mathcal{O}}

\newcommand{\T}{\ensuremath{\mathbb T}}

\newcommand{\vc}{\texttt{vc}}

\newcommand{\mm}{\texttt{mm}}

\newcommand{\true}{\textsf{true}}
\newcommand{\false}{\textsf{false}}

\newcommand{\CONPpoly}{\textsf{coNP/poly}}

\newcommand{\coNPpoly}{\textsf{coNP/poly}}

\newcommand{\yes}{\textsc{Yes}}
\newcommand{\no}{\textsc{No}}

\newtheorem{observation}{Observation}[section]
\newtheorem{reduction rule}{Reduction Rule}[section]
\newtheorem{marking-scheme}{Marking Scheme}[section]

\newcounter{condition}[section]

\newcommand{\defproblemout}[3]{
  \vspace{1mm}
\noindent\fbox{
  \begin{minipage}{0.96\textwidth}
  \begin{tabular*}{\textwidth}{@{\extracolsep{\fill}}lr} #1 \\ \end{tabular*}
  {\bf{Input:}} #2  \\
  {\bf{Output:}} #3
  \end{minipage}
  }
  \vspace{1mm}
}

\newcommand{\asg}{{\sf asg}}

\usepackage{hyperref}
\usepackage{graphicx}
\begin{document}
\title{Parameterized Complexity of \maxedgecoloring}
\titlerunning{Parameterized Complexity of \maxedgecoloring}
%
\author{Akanksha Agrawal\inst{1}
  \and
  Madhumita Kundu\inst{2}
  \and
  Abhishek Sahu\inst{3}
  \and
  Saket	Saurabh\inst{3, 4}	
  \and
  Prafullkumar Tale\inst{5}
}
\authorrunning{Agrawal et al.}
%
\institute{Ben Gurion University of the Negev, Israel. \email{agrawal@post.bgu.ac.il}
  \and
  Indian Statistical Institute, Kolkata, India. \email{kundumadhumita.134@gmail.com}
  \and
    The Institute of Mathematical Sciences, HBNI, Chennai, India.\\ 
    \email{\{asahu, saket\}@imsc.res.in}
  \and
  University of Bergen, Bergen, Norway.
  \and
Max Planck Institute for Informatics, Saarland Informatics Campus, Saarbr\"ucken, Germany. \email{prafullkumar.tale@mpi-inf.mpg.de}}
\maketitle              

\blfootnote{
\emph{Akanksha Agrawal:} Funded by the PBC Fellowship Program for Outstanding Post-Doctoral Researchers from China and India.

\noindent\emph{Saket Saurabh:} Funded by the European Research Council (ERC) 
under the European Union's Horizon $2020$ research and innovation programme (grant agreement No $819416$), and Swarnajayanti Fellowship (No DST/SJF/MSA01/2017-18).
\begin{minipage}{0.2\textwidth}
\includegraphics[scale=0.5]{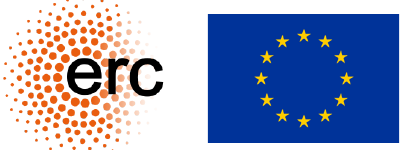}
\end{minipage}

\noindent\emph{Prafullkumar Tale:} Funded by the European Research Council (ERC) under the European Union’s Horizon 2020 research and innovation programme under grant agreement SYSTEMATICGRAPH (No. 725978). Most parts of this work was completed when the author was a Senior Research Fellow at The Institute of Mathematical Sciences, HBNI, Chennai, India.}

\begin{abstract}


A graph $H$ is {\em $p$-edge colorable} if there is a coloring $\psi: E(H) \rightarrow \{1,2,\dots,p\}$, such that for distinct $uv, vw \in E(H)$, we have $\psi(uv) \neq \psi(vw)$.
The {\sc Maximum Edge-Colorable Subgraph} problem takes as input a graph $G$ and integers $l$ and $p$, and the objective is to find a subgraph $H$ of $G$ and a  $p$-edge-coloring of $H$, such that $|E(H)| \geq l$.
We study the above problem from the viewpoint of Parameterized Complexity. We obtain \FPT\ algorithms when parameterized by: $(1)$ the vertex cover number of $G$, by using {\sc Integer Linear Programming}, and $(2)$ $l$, a randomized algorithm via a reduction to \textsc{Rainbow Matching}, and a deterministic algorithm by using color coding, and divide and color. With respect to the parameters $p+k$, where $k$ is one of the following: $(1)$ the solution size, $l$, $(2)$ the vertex cover number of $G$, and $(3)$ $l - {\mm}(G)$, where ${\mm}(G)$ is the size of a maximum matching in $G$; we show that  the (decision version of the) problem admits a kernel with $\mathcal{O}(k \cdot p)$ vertices. Furthermore, we show that there is no kernel of size $\mathcal{O}(k^{1-\epsilon} \cdot f(p))$, for any $\epsilon > 0$ and computable function $f$, unless $\NP \subseteq \CONPpoly$. 

\end{abstract}

\section{Introduction}
For a graph $G$, two (distinct) edges in $E(G)$ are \emph{adjacent} if they share an end-point. A $p$-edge coloring of $G$ is a function $\psi : E(G) \rightarrow \{1, 2, \dots, p\}$ such that adjacent edges receive different colors. One of the basic combinatorial optimization problems \textsc{Edge Coloring}, where for the given graph $G$ and an integer $p$, the objective is to find a $p$-edge coloring of $G$. \textsc{Edge Coloring} is a  very well studied problem in Graph Theory and Algorithm Design and we refer the readers to the survey by Cao et al. \cite{cao2019graph}, the recent article by Gr{\"u}ttemeier et al. \cite{gruttemeier2020maximum}, and references with-in for various known results, conjectures, and practical importance of this problem. 

The smallest integer $p$ for which $G$ is $p$-edge colorable is called its \emph{chromatic index} and is denoted by $\chi^{\prime}(G)$. The classical theorem of Vizing \cite{vizing1964estimate} states that $\chi^{\prime}(G) \le \Delta(G) + 1$, where $\Delta(G)$ is the maximum degree of a vertex in $G$. (Notice that by the definition of $p$-edge coloring, it follows that we require at least $\Delta(G)$ many colors to edge color $G$.) Holyer showed that deciding whether chromatic index of $G$ is $\Delta(G)$ or $\Delta(G) + 1$ is \NP-Hard even for cubic graphs \cite{holyer1981computational}.
Laven and Galil generalized this result to prove that the similar result holds for $d$-regular graphs, for $d \geq 3$~\cite{leven1983np}. 

\textsc{Edge Coloring} naturally leads to the question of finding the maximum number of edges in a given graph that can be colored with a given number of colors. This problem is called \maxedgecoloring\ which is formally defined below.

\defproblemout{\maxedgecoloring}{A graph $G$ and integers $l, p$}{A subgraph of $G$ with at least $l$ edges and its $p$-edge coloring or correctly conclude that no such subgraph exits.}

Note that the classical polynomial time solvable problem, {\sc Maximum Matching}, is a special case of \maxedgecoloring\ (when $p = 1$). Feige et al.~\cite{feige2002approximating} showed that \maxedgecoloring\ is \NP-hard even for $p = 2$. In the same paper, the authors presented a constant factor approximation algorithm for the problem and proved that for every fixed $p \geq 2$, there is $\epsilon > 0$, for which it is \NP-hard to obtain a $(1-\epsilon)$-approximation algorithm. Sinnamon presented a randomized algorithm for the problem \cite{sinnamon2019randomized}.
To the best of knowledge, Aloisioa and Mkrtchyan were the first to study this problem from the viewpoint of Parameterized Complexity \cite{aloisio2019fixed} (see Section~\ref{sec:prelims} for definitions related to Parameterized Complexity). Aloisioa and Mkrtchyana proved that when $p=2$, the problem is fixed-parameter tractable, with respect to various structural graph parameters like path-width, curving-width, and the dimension of cycle space. 
Gr{\"u}ttemeier et al.~\cite{gruttemeier2020maximum}, very recently, obtained kernels, when the parameter is $p+k$, where $k$ is one of the following: i) the number of edges that needs to be deleted from $G$, to obtain a graph with maximum degree at most $p - 1$,\footnote{Recall that any graph with maximum degree at most $p-1$, is $p$-edge colorable~\cite{vizing1964estimate}, and thus, this number is a measure of ``distance-from-triviality''.}, and ii) the deletion set size to a graph whose connected components have at most $p$ vertices. 
Galby et al.~\cite{galby2019parameterized} proved that \textsc{Edge Coloring} is fixed-parameter tractable when parameterized by the number of colors and the number of vertices having the maximum degree.

\vspace{0.1cm}
\noindent \textbf{Our Contributions:} 
Firstly, we consider \maxedgecoloring, parameterized by the vertex cover number, and we prove the following theorem.
\begin{theorem} \label{thm:fpt-vc}
\maxedgecoloring, parameterized by the vertex cover number of $G$, is \FPT.
\end{theorem}

We prove the above theorem, by designing an algorithm that, for the given instance, creates instances of \ilp, and the resolves the \ilp\ instance using the known algorithm~(\cite{kannan1987minkowski}, \cite{lenstra1983integer}). 
Intuitively, for the instance $(G, l, p)$, suppose $(H,\phi)$ is the solution that we are seeking for, and let $X$ be a vertex cover of $G$. (We can compute $X$ by the algorithm of Chen et al.~\cite{chen2010improved}.) We ``guess'' $H' = H[X]$ and $\phi' = \phi|_{E(H')}$. Once we have the above guess, we try to find the remaining edges (and their coloring), using \ilp. 

Next, we present two (different) \FPT\ algorithms for \maxedgecoloring, when parameterized by the number of edges in the desired subgraph, $l$. More precisely, we prove following theorem. 
\begin{theorem}
\label{thm:fpt-edges-colors}
There exists a deterministic algorithm $\calA$ and a randomized algorithm $\calB$ with constant probability of success that solves \maxedgecoloring. 
For a given instance $(G, l, p)$, 
Algorithms $\calA$ and $\calB$ terminate in time $\calO^*(4^{l + o(l)})$ and $\calO^*(2^l)$, respectively.
\end{theorem}
We remark that in the above theorem, the Algorithms $\calA$ and $\calB$ use different sets of ideas. 
Algorithm~$\calA$, uses a combination of the technique~\cite{chen2009randomized} of color-coding~\cite{DBLP:reference/algo/AlonYZ08} and divide and color. Algorithm~$\calB$ uses the algorithm to solve \textsc{Rainbow Matching} as a black-box.
We note that the improvement in the running time of Algorithm $\calB$ comes at the cost of de-randomization, as we do not know how to de-randomize Algorithm $\calB$.


Next we discuss our kernelization results. We show that (the decision version of) the problem admits a polynomial kernel, when parameteized by $p+k$, where $k$ is one of the following:
$(a)$~the solution size, $l$, $(b)$~the vertex cover number of $G$, and $(c)$~$l - {\mm}(G)$, where ${\mm}(G)$ is the size of a maximum matching in $G$; admits a kernel with $\mathcal{O}(kp)$ vertices. We briefly discuss the choice of our third parameter.
By the definition of edge coloring, each color class is a set of matching edges.
Hence, we can find one such color class, in polynomial time~\cite{micali1980v}, by computing a maximum matching in a given graph.
In above guarantee parameterization theme, instead of parameterizing, say, by the solution size ($l$ in this case), we look for some lower bound (which is the size of a maximum matching in $G$, for our case) for the solution size, and use a more refined parameter $(l - \mm(G))$. 
We prove the following theorem.
\begin{theorem} \label{thm:kernel} \maxedgecoloring\ admits a kernel with $\calO(kp)$ vertices, for every $k \in \{\ell, \vc(G), l - \mm(G)\}$. 
\end{theorem}

We complement this kernelization result by proving that the dependency of $k$ on the size of the kernel is optimal up-to a constant factor.
\begin{theorem}
\label{thm:kernel-lower-bound-solution-size}
For any $k \in \{\ell, \vc(G), l - \mm(G)\}$,  \maxedgecoloring\ does not admit a compression of size $\calO(k^{1-\epsilon} \cdot f(p))$, for any $\epsilon >0$ and computable function $f$, unless $\NP \subseteq \coNP/poly$.
\end{theorem} 

\section{Preliminaries}
\label{sec:prelims}

For a positive integer $n$, we denote set $\{1, 2, \dots, n\}$ by $[n]$.
We work with simple undirected graphs. The vertex set and edge set of a graph  $G$ are denoted as $V(G)$ and $E(G)$, respectively.
An edge between two vertices $u,v \in V(G)$ is denoted by $uv$.
For an edge $uv$, $u$ and $v$ are called its \emph{endpoints}.
If there is an edge $uv$, vertices $u, v$ are said to be {\em adjacent} to one another.
Two edges are said to be \emph{adjacent} if they share an endpoint.
The \emph{neighborhood} of a vertex $v$ is a collection of vertices which are adjacent to $v$ and it is represented as $N_G(v)$. The \emph{degree} of vertex $v$, denoted by $\text{deg}_G(v)$, is the size of its neighbhorhood.
For a graph $G$, $\Delta(G)$ denotes the maximum degree of vertices in $G$. 
The \emph{closed neighborhood} of a vertex $v$, denoted by $N_G[v]$, is the subset $N_G(v) \cup \{v\}$.
When the context of the graph is clear we drop the subscript.
For set $U$, we define $N(U)$ as union of $N(v)$ for all vertices $v$ in $U$.
For two disjoint subsets $V_1,V_2 \subseteq V(G)$, $E(V_1,V_2)$ is set of edges where one endpoint is in $V_1$ and another is in $V_2$. An edge in the set $E(V_1,V_2)$ is said to be {\em going across} $V_1, V_2$. 
For an edge set $E'$, $V(E')$ denotes the collection of endpoints of edges in $E'$.
A graph $H$ is said to be a \emph{subgraph} of $G$ if $V(H) \subseteq V(G)$ and $E(H) \subseteq E(G)$. In other words, any graph obtained from $G$ by deleting vertices and/or edges is called a subgraph of $G$. For a vertex (resp. edge) subset $X\ \subset \ V(G)$ (resp. $\ \subset \ V(G)$), $G - X$ ($G - Y$) denotes the graph obtained from $G$ by deleting all vertices in $X$ (resp. edges in $Y$). Moreover, by $G[X]$, we denote graph $G - (V(G) - X)$.

For a positive integer $p$, a $p$-edge coloring of a graph $G$ is a function $\phi: E(G) \rightarrow  \{1, 2, \dots, p\}$ such that for every distinct $uv$, $wx$ $\in E(G)$ s.t. $\left \{u,v  \right \} \cap  \left \{ w,x  \right \} \neq \emptyset$, we have $\phi(uv) \neq \phi(wx)$.
The least positive integer $p$ for which there exists a $p$-edge coloring of a graph $G$ is called \emph{edge chromatic number} of $G$ and it is denoted by $\chi'(G)$. 

\begin{proposition}[\cite{vizing1964estimate} Vizing]\label{thm:vizing}
For any simple graph $G$, $\Delta(G) \le \chi'(G) \le \Delta(G) + 1$.
\end{proposition}

For a coloring function $\phi$ and for any $i$ in $\{1, 2, \dots, p\}$, the edge subset $\phi^{-1}(i)$ is called the $i^{th}$ \emph{color class} of $\phi$. Notice that by the definition of $p$-edge coloring, every color class is a matching in $G$.
We define a \emph{balanced $p$-edge coloring} of a graph as a $p$-edge coloring in which the cardinality of any two color classes differ by at most one.

\begin{lemma}[\cite{feige2002approximating} Lemma $1.4$] \label{lemma:balanced-edge-coloring} For a graph $G$, let $\phi$ be a $p$-edge coloring of $G$. Then, there exists a balanced $p$-edge coloring of $G$ that can be derived from $\phi$ in polynomial time.
\end{lemma}

\begin{observation}\label{obs:color-to-matching} There exists a subgraph $H$ of $G$ such that $H$ is $p$-edge colorable and $|E(H)| \ge l$ if and only if there exists $p$ many edge disjoint matchings $M_1, M_2, \dots, M_p$ in $G$ such that $|M_1\cup M_2 \cup \cdots \cup M_p| \ge l$.
\end{observation}

For a graph $G$, a set of vertices $W$ is called an \emph{independent set} if no two vertices of $W$ are adjacent with each other.
A set $X \subseteq V(G)$ is a {\em vertex cover} of $G$ if $G - S$ is an independent set. The size of a minimum vertex cover of graph is called its \emph{vertex cover number} and it is denoted by $\vc(G)$.
A \emph{matching} of a graph $G$ is a set of edges of $G$ such that every edge
shares no vertex with any other edge of \emph{matching}. The size of maximum matching of a graph $G$ is denoted by $\mm(G)$. It is easy to see that $\mm(G) \le \vc(G) \le 2 \cdot \mm(G)$. 
\begin{definition}[\text{deg}-$1$-modulator] \label{def:deg-1-mod}
For a graph $G$, a set $X \subseteq V(G)$ is a {\em \text{deg}-$1$-modulator} of $G$, if the degree of each vertex in $G - X$  is at most 1.
\end{definition}

\paragraph{Expansion Lemma.}
Let \(t\) be a positive integer and $G$ be a bipartite graph with vertex bipartition $(P, Q)$.
A set of edges $M \subseteq E(G)$ is called a \emph{$t$-expansion of $P$ into $Q$} if (i) every vertex of $P$ is incident with exactly $t$ edges of $M$, and (ii) the number of vertices in \(Q\) which are incident with at least one edge in $M$ is exactly $t|P|$.
We say that \(M\) \emph{saturates} the end-points of its edges. 
Note that the set \(Q\) may contain vertices which are \emph{not} saturated by \(M\). 
We need the following generalization of Hall's Matching Theorem known as \emph{expansion lemmas}:

\begin{lemma}[See, for example, Lemma~$2.18$ in \cite{CFKLMPPS15}] \label{lem:expansion-lemma} Let $t$ be a positive
  integer and $G$ be a bipartite graph with vertex bipartition
  $(P,Q)$ such that $|Q| \geq t |P|$ and there are no isolated
  vertices in $Q$.  Then there exist nonempty vertex sets
  $P' \subseteq P$ and $Q' \subseteq Q$ such that (i) $P'$ has a
  $t$-expansion into $Q'$, and (ii) no vertex in $Q'$ has a
  neighbour outside $P'$. Furthermore two such sets $P'$ and $Q'$ can  be found in time polynomial in the size of $G$.
\end{lemma}


\paragraph{Integer Linear Programming.}
The technical tool we use to prove that  \maxedgecoloring\ is fixed-parameter tractable (defined in next sub-section) by the size of vertex cover is the fact that {\sc Integer Linear Programming} is fixed-parameter tractable when parameterized by the number of variables. An  instance of {\sc Integer Linear Programming} consists of a matrix $A \in \mathbb{Z}^{m \times q}$, a vector $\bar{b} \in \mathbb{Z}^m$ and a vector $\bar{c} \in \mathbb{Z}^q$. The goal is to find a vector $\bar{x} \in \mathbb{Z}^q$ which satisfies $A\bar{x} \le \bar{b}$ and minimizes the value of $\bar{c} \cdot \bar{x}$ (scalar product of $\bar{c}$ and $\bar{x}$). We assume that an input is given in binary and thus the size of the input instance or simply instance is the number of bits in its binary representation.

\begin{proposition}[\cite{kannan1987minkowski}, \cite{lenstra1983integer}] \label{prop:ilp} An \textsc{Integer Linear Programming} instance of size $L$ with $q$ variables can be solved using $\Oh(q^{2.5q+o(q)} \cdot (L + \log M_x ) \cdot \log(M_x \cdot M_c))$ arithmetic operations and space polynomial in $L + \log M_x$, where $M_x$ is an upper bound on the absolute value that a variable can take in a solution, and $M_c$ is the largest absolute value of a coefficient in the vector $\bar{c}$.
\end{proposition}

\paragraph{Parameterized Complexity.} 
\label{param}
The goal of parameterized complexity is to find ways of solving \NP-hard problems more efficiently than brute force by associating a {\em small} parameter to each instance. 
Formally, a {\em parameterization} of a problem is assigning a positive integer parameter $k$ to each input instance and we say that a parameterized problem is {\em fixed-parameter tractable } (\FPT) if there is an algorithm, that given an instance $(I,k)$, resolves in time bounded by $f(k)\cdot \vert I \vert ^{\mathcal{O}(1)}$, where $\vert I \vert$ is the size of the input $I$ and $f$ is an arbitrary computable function depending only on the parameter $k$. 

Such an algorithm is called an \FPT\ algorithm and such a running time is called \FPT\ running time.
Another central notion in the field of Parameterized Complexity is {\em kernelization}.
A parameterized problem is said to admit a $h(k)$-{\it kernel} if there is a polynomial-time algorithm (the degree of the polynomial is independent of $k$), called a {\em kernelization} algorithm, that, given an instance $(I,k)$ of the problem, outputs an instance $(I',k')$ of the problem such that: $(i)\ \left | I' \right | \in k' \leq h(k)$, and $(ii)$ $(I,k)$ and $(I',k')$ are {\em equivalent} instances of the problem i.e. $(I,k)$ is a {\sc Yes} instance if and only if $(I',k')$ is a {\sc Yes} instance of the problem.
It is known that a decidable problem admits an \FPT\ algorithm if and only if there is a kernel.
If the function $h(k)$ is polynomial in $k$, then we say that the problem admits a polynomial kernel. 
For more on parameterized complexity, see the recent books~\cite{CFKLMPPS15,fomin2019kernelization}.

We say a parameter $k_2$ is {\em larger} than a parameter $k_1$ if there exists a computable function $g(\cdot)$ such that $k_1 \le g(k_2)$. 
In such case, we denote $k_1 \preceq k_2$ and say $k_1$ is {\em smaller} than $k_2$. 
If a problem if \FPT\ parameterized by $k_1$ then it is also \FPT\ parameterized by $k_2$.
Moreover, if a problem admits a kernel of size $h(k_1)$ then it admits a kernel of size $h(g(k_2))$.
For a graph $G$, let $X$ be its minimum sized \text{deg}-$1$-modulator.
By the definition of vertex cover, we have $|X| \le \vc(G)$.
This implies $|X| \preceq \vc(G)$.
In the following observation, we argue that for ``non-trivial'' instances, $\vc(G) \preceq l$ and $|X| \preceq l - \mm(G)$. 

\begin{observation} \label{obs:para-hierarchy} For a given instance $(G, l, p)$ of \maxedgecoloring, in polynomial time, we can conclude that either $(G, l, p)$ is a \yes\ instance or $\vc(G) \preceq l$ and $|X| \preceq (l - \mm(G))$, where $X$ is a minimum sized \text{deg}-$1$-modulator of $G$. 
\end{observation}
\begin{proof}
Let $M_1$ be a maximum sized matching in graph $G$. 
Such matching can be found in polynomial time using the algorithm by Micali and Vazirani \cite{micali1980v}.
If $l \le |M_1|  = \mm(G)$ then we can conclude that $(G, l, p)$ is a \yes\ instance.  
Otherwise, we are working with an instance for which $l > \mm(G)$.
As $2 \mm(G) \ge \vc(G)$, we have $l > \vc(G)/2$  which implies $\vc(G) \preceq l$. 

Consider the graph $G' = G - M_1$. 
Let $M_2$ be a maximum sized matching in $G'$.
If $|M_1| + |M_2| \ge l$ then $(G, l, p)$ is a \yes\ instance.
Otherwise,  we are working with an instance for which $|M_2| < l - |M_1|$.
This implies $|V(M_2)| < 2(l - \mm(G))$. 
Consider the graph $G - V(M_2)$.
The only edges present in this graph are the ones in $M_1$.
Hence, every connected component in $G - V(M_2)$ has degree at most one.
This implies $|X| \le |V(M_2)| \le 2(l - \mm(G))$ where $X$ is a minimum sized \text{deg}-$1$-modulator of $G$.
\qed\end{proof}


\section{\FPT\ Algorithm Parameterized by the Vertex Cover Number of the Input}
\label{sec:fpt-vc}

In this section, we consider the problem \maxedgecoloring, when parameterized the vertex cover number of the input graph. Let $(G, l, p)$ be an instance of the problem, where the graph $G$ has $n$ vertices. We assume that $G$ has no isolated vertices as any such vertex is irrelevant for an edge coloring. We begin by computing a minimum sized vertex cover, $X$ of $G$, in time $\mathcal{O}(2^{|X|} n|X|)$, using the algorithm of Chen et al.~\cite{DBLP:conf/mfcs/ChenKX06}. 

We begin by intuitively explaining the working of our algorithm. We assume an arbitrary (but fixed) ordering over vertices in $G$, and let $W = V(G) \setminus X$. Suppose that we are seeking for the subgraph $H$, of $G$, with at least $\ell$ edges and the coloring $\phi: E(H) \rightarrow \{1, 2,\dots, p\}$. We first ``guess'' the intersection of $H$ with $G[X]$, i.e., the subgraph $H'$ of $G[X]$, such that $V(H) \cap X = H'$ and $V(H) \cap E(G[X]) = H'$. (Actually, rather than guessing, we will go over all possible such $H'$s, and do the steps, that we intuitively describe next.) Let $\phi' = \phi_{E(H')}$. Based on $(H', \phi')$, we construct an instance of \ilp, which will help us ``extend'' the partial solution $(H',\phi')$, to the solution (if such an extended solution exists), for the instance $(G,\ell, p)$. Roughly speaking, the construction of the \ilp\ relies on the following properties. Note that $W$ is an independent set in $G$, and thus edges of the solution that do not belong to $H'$, must have one endpoint in $X$ and the other endpoint in $W$. Recall that $H$ has the partition (given by $\phi$) into (at most) $p$ matchings, say, $M_1,M_2,\dots, M_{p'}$. The number of different neighborhoods in $X$, of vertices in $W$, is bounded by $2^k$. This allows us to define a ``type'' for $M_i - E(H')$, based on the neighborhoods, in $X$, of the vertices appearing in $M_i - E(H')$. Once we have defined these types, we can create a variable $Y_{\mathbb{T},\alpha}$, for each type $\mathbb{T}$ and color class $\alpha$ (in $\{0,1,\dots, p'\}$). The special color $0$ will be used for assigning all the edges that should be colored using the colors outside $\{1,2,\dots, p'\}$ (and we will later see that it is enough to keep only one such color). We would like the variable $Y_{\mathbb{T},\alpha}$ to store the number of matchings of type $\mathbb{T}$ that must be colored $\alpha$. The above will heavily rely on the fact that each edge in $H$ that does not belong to $H'$, must be adjacent to a vertex in $X$, this in turn will facilitate in counting the number of edges in the matching (via the type, where the type will also encode the subset of vertices in $X$ participating in the matching). Furthermore, only for $\alpha = 0$, the variable $Y_{\mathbb{T},\alpha}$ can store a value which is more than $1$. Once we have the above variable set, by adding appropriate constraints, we will create an equivalent instance of \ilp, corresponding to the pair $(H',\phi')$. We will now move to the formal description of the algorithm.

For $S \subseteq X$, let $\Gamma(S)$ be the set of vertices in $W$ whose neighborhood in $G$ is exactly $S$, i.e., $\Gamma(S) := \{w \in W|\ N_G(w) = S\}$. We begin by defining a tuple, which will be a ``type'', and later we will relate a matching (between $W$ and $X$), to a particular type. 

\begin{definition}[Type] \label{def:type}{\rm A {\em type} $\T = \langle X' = \{x_1, x_2, \dots, x_{|X'|}\}; S_1, S_2, \dots, S_{|X|}\rangle$ is a $(|X| + 1)$ sized tuple where each entry is a subset of $X$ and which satisfy following properties.
  \begin{enumerate}
  \item  The first entry, $X'$, is followed by $|X'|$ many entries which are non-empty subsets of $X$ and the remaining $(|X| + 1 - |X'|)$ entries are empty sets.
  \item Any non-empty set $S$ of $X$ appears at most $|\Gamma(S)|$ many times from the second entry onward in the tuple.  
  \item For every $i \in \{1, 2, \dots, |X'|\}$, we have $x_i  \in S_i$.
  \end{enumerate}
}\end{definition}

See Figure~\ref{fig:type-example} for an example. We note that the number of different types is at most $2^{|X|} \cdot 2^{|X|^2} \in 2^{\calO(|X|^2)}$ and it can be enumerated in time $2^{\calO(|X|^2)} \cdot n^{\calO(1)}$. We need following an auxiliary function corresponding to a matching, which will be useful in defining the type for a matching. Let $M$ be a matching across $X, W$ ($M$ has edges whose one endpoint is in $X$ and the other endpoint is in $W$). Define $\tau_M : X \cap V(M) \rightarrow W \cap V(M)$, as $\tau_M(x) := w$ if $xw$ is an edge in $M$. We drop when the context is clear.

\begin{definition}[Matching of type $\T$]\label{def:type-matching}{\rm A matching $M = \{x\tau(x) |\ x \in X \text{ and }$ $\tau(x) \in W\}$, is of {\em type} $\T = \langle X'; S_1, S_2, \dots, S_{|X|}\rangle$ if $V(M) \cap X = X'$ $(\ := \{x_1, x_2, \dots, x_{|X|'}\})$, and $S_j = N(\tau(x_j))$ for every $j$ in $\{1, 2, \dots, |X'|\}$. 
}\end{definition}

\begin{figure}[t]
  \begin{center}
    \includegraphics[scale=0.5]{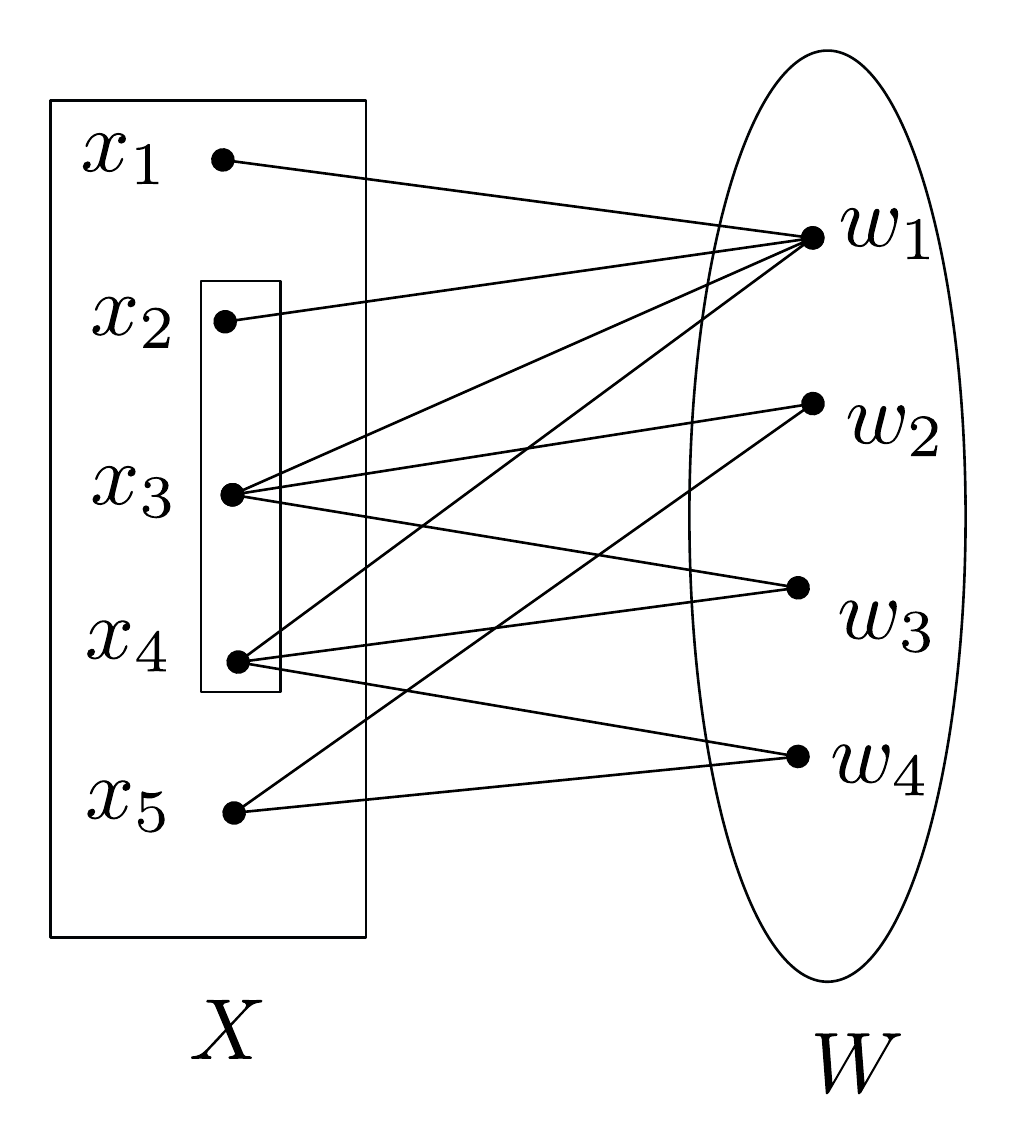}
  \end{center}
  \caption{The tuple $\T = \langle \{x_2, x_3, x_4\}; N(w_1), N(w_2), N(w_3), \emptyset, \emptyset \rangle$ is a \type. The matching $\{x_2w_1, x_3w_2, x_4w_3\}$ is of  $\T$. \label{fig:type-example}}
\end{figure}

We define some terms used in the sub-routine to construct an \ilp\ instance. 
For a type $\T = \langle X'; S_1, S_2, \dots, S_{|X|}\rangle$, we define $|\T| := |X'|$. Note that $|\T|$ is the number of edges in a matching of type $\T$. 
For a vertex $x \in X$ and a type $\T = \langle X'; S_1, S_2, \dots, S_{|X|}\rangle$, value of $\texttt{is\_present}(x, \T)$ is $1$ if $x\in X'$, and otherwise it is $0$.
For $w \in W$, define $\texttt{false\_twins}(w)$ as the number of vertices in $W$ which have the same neighborhood as that of $w$. That is, $\texttt{false\_twins}(w) = |\{\hat{w} \in W|\ N(w) = N(\hat{w}) \}|$. For a vertex $w\in W$ and a type $\T = \langle X'; S_1, S_2, \dots, S_{|X|}\rangle$, the value of $\texttt{nr\_nbr\_present}(w, \T)$ denotes the number of different $j$s in $\{1, 2, \dots, |X'|\}$ for which $S_j = N(w)$. We remark that the values of all the functions defined above can be computed in (total) time bounded by $2^{\calO(|X|^2)} \cdot n^{\calO(1)}$. 

\vspace{0.5cm}
\noindent \textbf{Constructing \ilp\ instances}. Recall that $G$ is the input graph and $X$ is a (minimum sized) vertex cover for $G$. Let $\mathfrak{T}$ be the set of all types. For every subgraph $H'$ of $G[X]$, a (non-negative) integer $p_0 \leq p$, and a $p_0$-edge coloring $\phi' : E(H') \rightarrow \{0, 1, 2, \dots, p_0\}$, we create an instance $I_{(H', \phi')}$, of \ilp\ as follows. Let $[p_0]' = \{0, 1, 2, \dots, p_0\}$. Define a variable $Y_{\T, \alpha}$ for every type $\T$ and integer $\alpha \in [p_0]'$. (These variables will be allowed to take values from $\{0, 1, \dots, p\}$). Intuitively speaking, for $\alpha$ in $[p_0]'$, the value assigned to $Y_{\T, \alpha}$ will indicates that there is a matching of type $\T$ which is assigned the color $\alpha$. Moreover, for $\alpha = 0$, the value of $Y_{\T, 0}$ will indicate that there are $Y_{\T, 0}$ many matchings of type $\T$, each of which must be assigned  a unique color which is strictly greater than $p_0$. Recall that for a type $\T \in \mathfrak{T}$, $|\T|$ is the number of edges in a matching of $ \T$. We next define our objective function, which (intuitively speaking) will maximize the number of edges in the solution. 
$$\text{maximize }\sum_{\T \in \mathfrak{T}; \alpha \in [p_0]'} Y_{\T, \alpha} \cdot |\T|$$
We next discuss the set of constraints. 

For every vertex $x$ in $X$, we add the following constraint, which will ensure that $x$ will be present in at most $p$ matchings: 
\begin{equation}
  \label{constraint:vc-side-vertex-capacity}
  \sum_{\T \in \mathfrak{T}; \alpha \in [p_0]'} Y_{\T, \alpha} \cdot \texttt{is\_present}(x, \T) \le p - \text{deg}_{H'}(x). \tag{${\sf ConstSetI}$}
\end{equation}

For each $x \in X$, an edge $x\hat{x}$ incident on $x$ in $H'$, and $\T \in \mathfrak{T}$, we add the following constraint, which will ensure that no other edge incident on $x$ and some vertex in $W$ is assigned the color $\phi'(x\hat{x})$: 
\begin{equation}
  \label{constraint:vc-side-vertex-coloring}
   Y_{\T, \phi'(x\hat{x})} \cdot \texttt{is\_present}(x, \T) = 0. \tag{{\sf ConstSetII}}
 \end{equation}

 We will next add the following constraint for each $w \in W$, which will help us in ensuring that $w$ is present in at most $p$ matchings: 
 
 \begin{equation}
   \label{constraint:ind-side-vetex-capacity}
   \sum_{\T \in \mathfrak{T}; \alpha \in [p_0]'} Y_{\T, \alpha} \cdot \texttt{nr\_nbr\_present}(w, \T) \le p \cdot \texttt{false\_twins}(w). \tag{{\sf ConstSetIII}}
 \end{equation}
 
  Notice that for two vertices $w_1, w_2 \in W$, such that $N(w_1) = N(w_2)$, the above constraints corresponding to $w_1$ and $w_2$ is exactly the same (and we skip adding the same constraint twice).  
  
When $\alpha \neq 0$, we want to ensure that at most one matching that is colored $\alpha$. Thus, for $\alpha \in [p_0]$, add the constraint: 
 \begin{equation}
   \label{constraint:type-used-in-VC}
   \sum_{\T \in \mathfrak{T} } Y_{\T, \alpha} \le 1. \tag{{\sf ConstSetIV}} 
 \end{equation}

Note that we want at most $p$ color classes, which will be ensured by our final constraint as follows.
\begin{equation}
  \label{constraint:nr-color-class}
  \sum_{\T \in \mathfrak{T}; \alpha \in [p_0]'} Y_{\T, \alpha} \le p. \tag{{\sf ConstSetV}} 
\end{equation}

This completes the construction of the \ilp\ instance of $I_{(H',\phi')}$. 
\vspace{0.5cm} 

\vspace{0.5cm}
 \noindent \textbf{Algorithm for \maxedgecoloring:} Consider the given instance $(G, l, p)$ of \maxedgecoloring. The algorithm will either return a solution $(H, \phi)$ for the instance, or conclude that no such solution exists. We compute a minimum sized vertex cover, $X$ of $G$, in time $\mathcal{O}(2^{|X|} n|X|)$, using the algorithm of Chen et al.~\cite{DBLP:conf/mfcs/ChenKX06}. For every subgraph $H'$ of $G[X]$, a (non-negative) integer $p_0 \leq p$, and a $p_0$-edge coloring $\phi' : E(H') \rightarrow \{0, 1, 2, \dots, p_0\}$, we create the instance $I_{(H', \phi')}$, and resolve it using Proposition~\ref{prop:ilp}. (In the above we only consider those $\phi': E(H') \rightarrow \{0, 1, 2, \dots, p_0\}$, where each of the color classes are non-empty.) If there exists a tuple $(H', \phi')$ for which the optimum value of the corresponding \ilp\ instance is at least $(l - |E(H')|)$ then algorithm constructs a solution $(H,\phi)$ as specified in the proof of Lemma~\ref{lemma:backward-direction} and returns it as a solution. If there is no such tuple then the algorithm concludes that no solution exists for a given instance.
\vspace{0.5cm} 

For a solution $(H, \phi: E(H) \rightarrow [p])$ for the instance $(G, l, p)$, we say that $(H,\phi)$ is a {good solution}, if for some $p_0 \in [p]$, for each $e \in E(H) \cap E(G[X])$, we have $\phi(e) \in [p_0]$. Note that if $(G, l, p)$ has a solution, then it also has a good solution. We argue the correctness of the algorithm in the following two lemmas. 

\begin{lemma}
  \label{lemma:forward-direction}  If $(G, l, p)$ has a good solution $(H, \phi)$ then the optimum value of the \ilp\ instance $I_{(H',\phi')}$ is at least $(l - |E(H')|)$, where $H' = H[X]$ and $\phi' : E(H') \rightarrow \{1,2,\dots, p_0\}$, such that $\phi'= \phi|_{E(H')}$ and $p_0 = \max\{\phi(e) \mid e \in E(H) \cap E(G[X])\}$. 
 \end{lemma}
 \begin{proof}
   Let $M_1, M_2, \dots, M_p$ be the partition of edges in $E(H) \setminus E(H')$ according to the colors assigned to them by $\phi$, and ${\cal M} = \{M_i \mid i\in [p]\} \setminus \{\emptyset\}$. Notice that each $M_i$ is a matching, where the edges have one endpoint in $X$ and the other endpoint in $W$. We create an assignment $\asg: {\sf Var}_{(H',\phi')} \rightarrow [p_0]'$, where ${\sf Var}_{(H',\phi')}$ is the set of variables in the instance $I_{(H',\phi')}$ as follows.  Initialize $\asg(z) = 0$, for each  $z\in {\sf Var}(I_{(H',\phi')})$. For $i \in [p]$, let $\T_i$ be the type of $M_i$ and $p_i = \phi(e)$, where $e \in M_i$. For each $i \in [p]$, we do the following. If $p_i > p_0$, then increment $\asg(Y_{\T_i, 0})$ by one, and otherwise increment value of $\asg(Y_{\T_i, p_i})$ by one. This completes the assignment of variables. Next we argue that $\asg$ satisfies all constraints in $I_{(H',\phi')}$ and the objective function evaluates to a value that is at least $(l - |E(H')|)$. 

As there are at most $p$ matchings, we have $ \sum_{\T \in \mathfrak{T}; \alpha \in [p_0]'} Y_{\T, \alpha} \le p$, and thus, the constraint in {\sf ConstSetV} is satisfied.
 
We will now argue that each constraint in {\sf ConstSetI} is satisfied. To this end, consider a variable $x \in X$, and let ${\sf a}_x = \sum_{\T \in \mathfrak{T}; \alpha \in [p_0]'} \asg(Y_{\T, \alpha}) \cdot \texttt{is\_present}(x, \T)$. Since $H$ is $p$-edge colorable, $\text{deg}_H(x) \le \Delta(H) \le p$. Hence, there are at most $p$ edges incident on $x$ in $H$ (Proposition~\ref{thm:vizing}). For any $\T \in \mathfrak{T}$ and $\alpha \in [p_0]'$, if $\asg(Y_{\T, \alpha}) \cdot \texttt{is\_present}(x, \T) \neq 0$, then there are $\asg(Y_{\T, \alpha})$ many matchings of type $\T$ in ${\cal M}$, each of which contains an edge incident on $x$. Moreover, each such matching contains a different edge incident on $x$. 
Since $\phi$ is a $p$-edge coloring of $H$, we have ${\sf a}_x +\ \text{deg}_{H'}(x) =\ \text{deg}_{H}(x) \le p$. This implies that ${\sf a}_x = \sum_{\T \in \mathfrak{T}; \alpha \in [p_0]'} \asg(Y_{\T, \alpha}) \cdot \texttt{is\_present}(x, \T) \le p - \text{deg}_{H'}(x)$. Thus we conclude that all contraints in ${\sf ConstSetI}$ are satisfied. 

Now we argue that all constraints in {\sf ConstSetII} are satisfied. Consider $x \in X$, an edge $x\hat{x}$ incident on $x$ in $H'$, and $\T \in \mathfrak{T}$ such that $\texttt{is\_present}(x, \T) = 1$. Since $x\hat{x} \in E(H')$, there is no matching $M_i \in {\cal M}$, such that $p_i = \phi'(x\hat{x})$ and $M$ contains an edge incident on $x$. Thus we can obtain that $\asg(Y_{\T, \phi'(x\hat{x})}) = 0$ (recall that $\texttt{is\_present}(x, \T) = 1$). From the above we can conclude that $\asg(Y_{\T, \phi'(x\hat{x})}) \cdot \texttt{is\_present}(x, \T) = 0$.

Next we argue that all constraints in {\sf ConstSetIII} are satisfied. To this end, consider a (maximal) subset $W' = \{w_1, w_2, \dots, w_r \} \subseteq W$, such that any two vertices in $W'$ are false twins of each other. Notice  that for each $j,j' \in [r]$, $\sum_{\T \in \mathfrak{T}; \alpha \in [p_0]'} Y_{\T, \alpha} \cdot \texttt{nr\_nbr\_present}(w_j, \T) \le p \cdot \texttt{false\_twins}(w_j)$ is exactly the same as $\sum_{\T \in \mathfrak{T}; \alpha \in [p_0]'}$ $Y_{\T, \alpha} \cdot \texttt{nr\_nbr\_present}(w_{j}, \T) \le p \cdot \texttt{false\_twins}(w_{j'})$.  
Consider any $w \in W'$, $\T \in \mathfrak{T}$, and $\alpha \in [p_0]'$, such that we have $\asg(Y_{\T, \alpha}) \cdot \texttt{nr\_nbr\_present}(w, \T) \neq 0$. There are $\asg(Y_{\T, \alpha})$ many matchings in $\cal M$ each of which contains $\texttt{nr\_nbr\_present}(w, \T)$ many edges incident vertices in $W'$. Hence $\sum_{\T \in \mathfrak{T}; \alpha \in [p_0]'} \asg(Y_{\T, \alpha}) \cdot \texttt{nr\_nbr\_present}(w_j, \T)$ is the number of edges incident on $W'$ in $H$. Note that $p \cdot \texttt{false\_twins}(w)$ is the maximum number of edges in $H$ which can be incident on vertices in $W'$. Thus we can conclude that $\sum_{\T \in \mathfrak{T}; \alpha \in [p_0]'} \asg(Y_{\T, \alpha}) \cdot \texttt{nr\_nbr\_present}(w_j, \T) \leq p \cdot \texttt{false\_twins}(w)$. 

For any $\alpha \in [p_0]$, there is at most one matching in $\cal M$ whose edges are assigned the color $\alpha$. This implies that $\sum_{\T \in \mathfrak{T} } \asg(Y_{\T, \alpha}) \le 1$. Hence all constraints in {\sf ConstSetIV} are satisfied. 

There are at least $(l - |E(H')|)$ many edges in $E(H) \setminus E(H')$ and each such edge has one endpoint in $X$ and another in $W$. Every edge in matching contributes exactly one to the objective function. Thus we can obtain that $\sum_{\T \in \mathfrak{T}; \alpha \in [p_0]'} \asg(Y_{\T, \alpha}) \cdot |\T| \geq (l - |E(H')|)$. This concludes the proof.
%
\qed\end{proof}
 
 \begin{lemma}
   \label{lemma:backward-direction}
If there is $(H', \phi')$ for which the optimum value of the \ilp\ instance $I_{(H', \phi')}$, is at least $(l - |E(H')|)$, then the \maxedgecoloring\ instance $(G, l, p)$ admits a solution. Moreover, given $\asg: {\sf Var}_{(H',\phi')} \rightarrow [p_0]'$, where ${\sf Var}_{(H',\phi')}$, we can be compute $(H, \phi)$ in polynomial time. 
 \end{lemma}
 \begin{proof} 
 We first describe an algorithm, which given an assignment $\asg: {\sf Var}_{(H',\phi')} \rightarrow [p_0]'$ for $I_{(H', \phi')}$, such that the optimum value of objective functions is at least $(l - |E(H')|)$, constructs a solution $(H, \phi)$ for $(G, l, p)$. We will construct $(H, \phi)$, such that $(1)$ $E(H') \subseteq E(H)$, $(2)$ $\phi$ is a $p$-edge coloring of $H$ which has at least $l$ edges, and $(3)$ $\phi|_{E(H')}$ is identical to that of $\phi'$. For every variable $Y_{\T, \alpha} \in {\sf Var}_{(H',\phi')}$, such that $\asg(Y_{\T, \alpha}) \neq 0$, the algorithm will constructs a matching $M_\alpha$ with $|\T|$ edges. At each step, the edges in $M_\alpha$ are added to $H$, and $\phi$ assigns the color $\alpha$ to all the edges in $M_\alpha$. We will argue that by the end of this process, the number of edges in $H$ is at least $l$.
Recall that there is a fixed ordering on vertices in $X$ and $W$. 
and for a subset $S_i$ of $X$, $\Gamma(S_i)$ denotes the collection of vertices in $W$ whose neighborhood in $G$ is exactly $S_i$. 
We say that $\Gamma(S_i)$ is \emph{$H$-degree balanced set} if for any two vertices $w_1, w_2$ in the set, $\text{deg}_H(w_1)$ and $\text{deg}_H(w_2)$ differs by at most one. 

\vspace{0.3cm}
\noindent \textbf{Algorithm to construct $(H, \phi)$ :} Initialize $V(H) = V(G)$, $E(H) = E(H')$, and $\phi|_{E(H')} = \phi'$. 
Consider $\alpha > 0$ and a type $\T = \langle X'; S_1, S_2, \dots, S_{|X|}\rangle$ for which $\asg(Y_{\T, \alpha}) \neq 0$.
For the sake of clarity, assume $X' = \{x_1, x_2, \dots, x_{|X'|}\}$.
The algorithm constructs a matching $M_{\T,\alpha}$, of type $\T$ in the following way.
Initialize $M_{\T,\alpha} = \emptyset$.
For $i$ in $\{1, 2, \dots, |X'|\}$, let $w^i$ be a vertex in $\Gamma(S_i)$ such that $(a)$ no edge incident on $w^i$ has already been added to $M_{\T,\alpha}$, $(b)$ degree of $w^i$ in $H$ is at most $p-1$, and $(c)$ $\Gamma(S_i)$ remains $H$-degree-balanced after increasing degree of $w^i$ by one.
If there are more than one vertices that satisfy these properties, select the lowest indexed vertex as $w^i$.
Add edge the edge $x_iw^i$ to $M_{\T,\alpha}$ before moving to next value of $i$. This completes the construction of $M_{\T,\alpha}$. 
Add all the edges in $M_{\T,\alpha}$ to $H$ and assign $\phi(e) = \alpha$, for every edge $e$ in $M$. 

Now we will consider variables $Y_{\T, 0}$, such that $\asg(Y_{\T, 0}) \neq 0$, for $\T \in \mathfrak{T}$. Set $\alpha_0 := p_0$. Consider $\T \in \mathfrak{T}$, and let $\asg(Y_{\T, 0}) = {\sf a}_\T$. 
For each (increasing) $\beta \in [{\sf a}_\T]$, do the following. We construct a matching $M'_{\T,\beta}$ of type $\T$ similar to the one that we diiscussed earlier. That is, initialize $M'_{\T,\beta} = \emptyset$.
For $i$ in $\{1, 2, \dots, |X'|\}$, let $w^i$ be a vertex in $\Gamma(S_i)$ such that $(a)$ no edge incident on $w^i$ has already been added to $M_{\T,\beta}$, $(b)$ degree of $w^i$ in $H$ is at most $p-1$, and $(c)$ $\Gamma(S_i)$ remains $H$-degree-balanced after increasing degree of $w^i$ by one.
If there are more than one vertices that satisfy these properties, select the lowest indexed vertex as $w^i$. Add all the edges in $M'_{\T,\beta}$ to $H$, set $\phi(e) = \alpha_0 + \beta$, and move to the next choice of $\beta$ (if it exists). This completes the description of the algorithm. 

It is clear from the description of the algorithm that it can be executed in polynomial time. We next argue that: i) for each $\alpha \in [p]$, the algorithm constructs $M_{\T,\alpha}$ with $|\T|$ edges, for whenever $\asg(Y_{\T,\alpha}) = 1$, and ii) for each $\beta \in [{\sf a}_\T]$ the algorithm constructs $M_{\T,\beta}$ with $|\T|$ edges, whenever $\asg(Y_{\T, 0}) \neq 0$. We argue only the first statement, the proof of the second statement can be obtained by following similar arguments. For the sake of contradiction, assume that there is $\alpha \in [p]$ and $\T = \langle X'; S_1, S_2, \dots, S_{|X|}\rangle$, for which the algorithm could not construct $M_{\T,\alpha}$ of size $\T$, or in other words, by construction, the algorithm could not construct $M_{\T,\alpha}$ of type $\T$. In the above, we consider the lowest iteration under which $(\alpha,\T)$ was under consideration and $M_{\T,\alpha}$ of type $\T$ could not be constructed. Thus, for some $x_i \in X'$, for every vertex $w \in \Gamma(S_i)$ at least one of the following holds: $(a)$ $w$ has an edge incident on it which has already been added to $M_{\T,\alpha}$, $(b)$ $w$ has degree exactly $p$ in $H$, or $(c)$ $\Gamma(S_i)$ does not remains an $H$-degree-balanced after increasing the degree of $w$ in $H$ by one.
We consider following two exhaustive cases: Case~$(1)$ There exists a vertex in $\Gamma(S_i)$ whose degree in $H$ is $p$.
Case~$(2)$ Every vertex in $\Gamma(S_i)$ has degree at most $(p - 1)$ in $H$. We argue that Case~$(1)$ leads to the contradiction that the constraints in {\sf ConstSetIII} is satisfied. We argue that in Case~$(2)$, $\T$ is not a type, again leading to a contradiction. 

Consider Case~$(1)$. Since the algorithm failed for the first time, $\Gamma(S_i)$ is an $H$-degree balanced set and each vertex in it has a degree at most $p$, before the algorithm started processing for the iteration for $\alpha$ and $\T$. As $\Gamma(S_i)$ at the current processing contains a vertex of degree $p$ in $H$, every vertex in it has degree either $(p - 1)$ or $p$ in $H$. Suppose there are $n_0$ vertices in $\Gamma(S)$ which have degree $(p - 1)$. 
Let $w$ be a vertex in $\Gamma(S)$. By definition, we have $\texttt{false\_twins}(w) = |\Gamma(S)|$.
Let $L$ denote the summation on the left hand side of the constraint of type~(\ref{constraint:ind-side-vetex-capacity}) corresponding to $w$. Let $L'$ be the summation of $\asg(Y_{\T', \alpha})$ which have already been processed by the algorithm. Note that $L \ge L' + \asg(Y_{\T, \alpha}) \cdot \texttt{nr\_nbr\_present}(w, \T)$.
Recall that for $ \T' \in \mathfrak{T}$ and integer $\alpha'$ if both $\asg(Y_{\T', \alpha'})$ and $\texttt{nr\_nbr\_present}(w, \T')$ are non-zeros, the algorithm adds $\asg(Y_{\T', \alpha'}) \cdot \texttt{nr\_nbr\_present}(w, \T')$ many edges incident on vertices in $\Gamma(S_i)$.
Since there are $n_0$ vertices with degree $(p - 1)$ and $\Gamma(S_i)$ many vertices of degree $p$ in $H$, there are $p \cdot (|\Gamma(S_i)| - n_0) + (p - 1) \cdot n_0$ many edges incident on vertices in $\Gamma(S_i)$. As $\Gamma(S_i)$ is a subset of $W$, which is an independent set in $G$ and hence in $H$, all these edges are across $X, W$. Hence, algorithm has added $L' = p \cdot |\Gamma(S_i)| + n_0\cdot (p - 1)$ many edges to $H$ before it starts processing at $ \alpha$, $\T$. Since $n_0$ many edges, each of which incident on vertices in $\Gamma(S_i)$ which has degree $(p - 1)$, is not sufficient to construct a matching of $ \T$, we have $\asg(Y_{\T, \alpha}) \cdot \texttt{nr\_nbr\_present}(w, \T) \ge n_0 + 1$.
This implies $L \ge L' + \asg(Y_{\T, \alpha}) \cdot \texttt{nr\_nbr\_present}(w, \T) \ge p \cdot (|\Gamma(S_i)| - n_0) + (p - 1) \cdot n_0 + n_0 + 1 \ge p \cdot |\Gamma(S_i)| + 1$. As $|\Gamma(S_i)| = \texttt{false\_twins}(w)$, this contradicts the fact that constraint in {\sf ConstSetIII} are satisfied.

Consider Case~$(2)$.
Before the algorithm starts processing $\alpha$, $\T$,  
every vertex has degree at most $(p - 1)$ and $\Gamma(S_i)$ is an $H$-degree balanced set. The algorithm can select one edge incident on every vertex in $\Gamma(S_i)$ to add it to the matching. 
Note that the algorithm selects an edge incident on vertices in $\Gamma(S_i)$ if and only if the entry in the tuple is $S_i$.
Since the algorithm failed in this case, we can conclude that
$S_i$ appears at least $|\Gamma(S_i)| + 1$ many times from second place onward in $\T$. This contradicts the second property mentioned in Definition~\ref{def:type}.

As discussed in the previous two paragraphs, both Case~$(1)$ and Case~$(2)$ lead to contradictions. Hence our assumption that the algorithm is not able to construct a matching at certain steps is wrong. This implies the algorithm will always return $(H, \phi)$.


We now argue that $\phi$ is an edge coloring of $H$. Consider an arbitrary vertex $x$ in $X$. Consider an edge $x\hat{x}$ in $H'$ which is incident on $x$. By the construction of {\sf ConstSetII}, for any $\T$ if $\texttt{is\_present}(x, \T) = 1$, then $\asg(Y_{\T, \phi'(x\hat{x})})=0$. Hence, at no step the algorithm modifies $\phi$ in a way that it assign color $\phi(x\hat{x})$ to a newly added edge which is incident on $x$.
Moreover, by the constraints in {\sf ConstSetIV}, for $\alpha > 0$ and $\T$ if $\asg(Y_{\T, \alpha})\neq 0$, then for any other $\T' \in \mathfrak{T}$, $\asg(Y_{\T', \alpha}) = 0$. Hence, the algorithm does not add more that one edge of color $\alpha$ on any vertex in $x$.
Consider an arbitrary vertex $w$ in $W$. At the start of the process, there is no edge incident on $w$.
At any stage, the algorithm adds at most one edge to $H$ and assigns it a color that has not been used previously and will not be used later. Hence, every edge incident on $w$ has been assigned to a different color. As $x, w$ are arbitrary vertices in $X, W$, respectively, we can conclude that $\phi$ is an edge coloring of $H$. 

We argue that $\phi$ uses at most $p$ colors. 
Consider a vertex $x$ in $X$. Because of the constraints in {\sf ConstSetI}, the algorithm adds at most $p - \text{deg}_{H'}(x)$ many edges incident on $x$. Hence there are at most $p$ edges incident on any vertices in $X$.
Consider a vertex $w$ in $W$. As mentioned before, there are no edges incident on $w$ at the start of the process. By the constraint in {\sf ConstSetV}, the process creates at most $p$ matchings. Hence there are at most $p$ many edges incident on $w$.
As $x, w$ are arbitrary vertices in $X, W$, respectively, we can conclude that $\phi$ is a $p$-edge coloring of $H$. 

The algorithm adds $Y_{\T, \alpha} \cdot |\T|$ many edges for every variable which has non-zero value. Since the objective function is at least $(l - |E(H')|)$, we can conclude that $H$ has at least $l$ edges. This concludes the proof of the lemma.
\qed\end{proof}

We are now in a position to state the main result of this section.

\begin{proof}[Proof of Theorem~\ref{thm:fpt-vc}]
We prove that there exists an algorithm which given a graph $G$ on $n$ vertices and integers $l, p$ as input either outputs a subgraph $H$ of $G$ such that $H$ is $p$-edge colorable and has at least $l$ edges, or correctly concludes that no such subgraph exists. Moreover, the algorithm terminates in time $f(\vc(G)) \cdot n^{\calO(1)}$, where $f(\vc(G))$ is some computable function which depends only on $\vc(G)$.

  We argue that the algorithm described in this section satisfy desired properties. The correctness of the algorithm is implied by Lemma~\ref{lemma:forward-direction} and Lemma~\ref{lemma:backward-direction}. We now argue that the algorithm runs in \FPT\ time.
  The algorithm computes an optimum vertex cover $X$ in time $2^{\vc(G)} \cdot |V(G)|^{\calO(1)}$.
It then enumerates all tuples of type $(H', \phi')$ where $H'$ is a subgraph of $G[X]$ and $\phi'$ is a $p_0$-edge coloring of $H'$ for some $p_0 \le p$. There are $2^{\calO(|X|^2)}$ many possible choices for $H'$. Recall that the algorithm only considers $\phi'$ in which for every $j$ in $\{1, 2, \dots, p_0\}$ there is an edge $x\hat{x}$ in $E(H')$ such that $\phi'(x\hat{x}) = j$. Hence $p_0 \le |X|^2$. This implies that the total number of choices for $\phi$ is at most $|X|^{\calO( |X|^2)}$. Hence the algorithm creates at most $2^{\calO(|X|^2 \log |X|)}$ many instances of \ilp. 

The algorithm uses Proposition~\ref{prop:ilp} to solve each instance of \ilp. To bound the time taken for this step, we bound the number of variables in each instance of \ilp.
As mentioned earlier, the number of different types is at most $2^{|X|} \cdot 2^{|X|^2} \in 2^{\calO(|X|^2)}$ and all of them can be enumerated in time $2^{\calO(|X|^2)} \cdot n^{\calO(1)}$. Since, $\alpha$ can have at most $p_0 + 1 \le \calO(|X|^2)$ different values, every instance has $2^{\calO(|X|^2)}$ many variables. By construction, upper bounds on the absolute value a variable can take in a solution and the largest absolute value of a coefficient used is linearly bounded by $n$. By Proposition~\ref{prop:ilp}, this instance can be solved in time $2^{2^{{\calO(|X|^2)}}}\cdot n^{\calO(1)}$. Hence the algorithm terminates in time \FPT\ in $\vc(G)$ which concludes the proof. 
\qed\end{proof}
\section{An \FPT\ Algorithm Parameterized by the Number of Edges in a Desired Subgraph}
\label{sec:fpt-max-edges}

In this section, we prove Theorem~\ref{thm:fpt-edges-colors}. We say a randomized algorithm $\calB$ \emph{solves} \maxedgecoloring\ problem with constant probability of success if given an instance $(G, l, p)$ such that $G$ contains a subgraph $H$ which is $p$-edge colorable and $|E(H)| \ge l$, the algorithm returns a solution with constant probability.
Our first algorithm uses the technique of color-coding combined with divide and color introduced in \cite{chen2009randomized}. 
We present a randomized version of this algorithm which can be de-randomized using standard techniques (see for example~\cite{CFKLMPPS15}). For the second algorithm, we reduce a given instance of \maxedgecoloring\ to an equivalent instance of \textsc{Rainbow Matching}.
This reduction along with the known algorithm for the later problem results in a different randomized \FPT\ algorithm for \maxedgecoloring, with improved running time.

\subsection{A Deterministic \FPT\ Algorithm}
\label{sub-sec:fpt-edges-color}
\DeclarePairedDelimiter\ceil{\lceil}{\rceil}
\DeclarePairedDelimiter\floor{\lfloor}{\rfloor}

Given an instance $(G,l,p)$ of \maxedgecoloring\ problem, we can assume $l \equiv 0 \pmod{p}$. 
If it is not the case, then let $l \equiv r \pmod{p}$ for some $r \in [p-1]$. 
We create another instance $(G', l'=l+(p-r), p)$ where $G'$ is the graph obtained obtained by adding $(p - r)$ isolated edges. 
Formally,  $V(G')=V(G)\cup \{x_1,x_2,\dots,x_{2(p-r)}\}$ and $E(G')= E(G) \cup \{ x_{2i-1}x_{2i}|\ i \in \{1, 2, \dots, p - r\}\}$.
It is easy to see that $(G', l + (p - r), p)$ is a \yes\ instance if and only if $(G, l, p)$ is a \yes\ instance.
By  Lemma~\ref{lemma:balanced-edge-coloring},  if $(G',l',p)$ is a \yes\ instance of \maxedgecoloring\ problem, then there is a $p$-edge-coloring of $G'$ where exactly $q=l'/p$ edges are colored by every color.
Hence, in the remaining section, we assume that for a given instance $(G, l, p)$, we have $l \equiv 0 \pmod{p}$.

We present a randomized recursive algorithm (Algorithm~\ref{alog:faster-randomized-algo}) to solve the problem and later specify how to de-randomize it. The central idea is to partition the edge set into two parts such that one part contains all the solution edges colored by 
 the first $\floor{a/2}$ colors and the other part contains all the solution edges colored by the remaining  $\ceil{a/2}$ colors.
 We determine the answer to these subproblems recursively and use them to return the answer to the original problem. 
To formalize these ideas,  we define the term ${D}_{(a,q)}[X]$ for $a\in \mathbb{N}$, $X\subseteq E(G')$, where ${D}_{(a,q)}[X]$  is \true\ if and only if there are $a$ edge-disjoint matchings, each of size $q$, in $X$.
 Instead of computing ${D}_{(a,q)}[X]$, the algorithm computes  ${D}^\star_{(a,q)}[X]$.
 The relationship between these terms is as follows:  if ${D^\star}_{(a,q)}[X]$ is \true\ then ${D}_{(a,q)}[X]$ is always \true, but if ${D}_{(a,q)}[X]$ is \true\  then ${D^\star}_{(a,q)}[X]$ is \true\ only with sufficiently high probability.
Thus, we get a one-sided error Monte Carlo algorithm.
We boost the success probability of correct partitions by repeating the partitioning process many times, to achieve constant success probability.
We note that the fact that each color class contains exactly $q$ many edges ensures that at each partitioning step, two parts contain an almost equal number of edges.
This fact plays a crucial role while calculating the probability of success and the run time of the algorithm.

\begin{algorithm}[t]
	\SetKwInOut{Input}{Input}
 	\SetKwInOut{Output}{Output}
	\Input{A subset $X\subseteq E(G)$, integers $1\leq a\leq p$ and $q$. }
	\Output{ $D_{a, q}[X]$ ($D_{a, q}[X]$ is \true\ if and only if there are $a$ edge disjoint matchings, each of size $q$, in $G[X]$)}
	\If{$a == 1$}{
		\Return \true\ if there is a matching of size $q$ in $G[X]$, and otherwise \false.
	}
	${D^\star}_{(a,q)}[X] = $ \false\;
	\For{$2^{aq}\log{(4l)}$ many times}{
	 Partition $X$ into $L \uplus R$ uniformly at random\;
	${D^\star}_{(\floor{a/2},q)}[L]=$ Faster-Randomized-Algorithm($L, \floor{a/2}, q$)\;
	${D^\star}_{(\ceil{a/2},q)}[R]=$ Faster-Randomized-Algorithm($R, \ceil{a/2}, q$)\;
	\If{${D^\star}_{(a,q)}[X]==$  \false}{
		${D^\star}_{(a,q)}[X]={D^\star}_{(\floor{a/2},q)}[L] \wedge {D^\star}_{(\ceil{a/2},q)}[R]$\;
	}
	}
	\Return ${D^\star}_{(a,q)}[X]$
 \caption{Faster-Randomized-Algorithm$(X, a, q)$ \label{alog:faster-randomized-algo}}
\end{algorithm}

\begin{lemma} \label{lemma:fpt-edges-colors-random-1}
There exists a randomized algorithm that given $(G, l, p)$ either finds a subgraph $H$ of $G$ and its $p$-edge coloring such that $|E(H)| \ge l$, or correctly concludes that no such subgraph exists in time $\calO^*(4^{l + o(l + p)})$.
Moreover, if such a subgraph exists in $G$, then the algorithm returns it with constant probability. 
\end{lemma}
\begin{proof}
Given an instance $(G, l, p)$, the algorithm does the necessary modifications (as mentioned in the starting of this sub-section) to ensure that $l \equiv 0 \pmod{p}$.
It then runs Algorithm~\ref{alog:faster-randomized-algo} with $X = E(G)$, $a = p$, and $q = l/p$ as input. 
If Algorithm~\ref{alog:faster-randomized-algo} return \true\ then the algorithm returns \yes\ otherwise it returns \no.
It is easy to modify Algorithm~\ref{alog:faster-randomized-algo}, and hence the algorithm, to ensure that the algorithm returns a set of edges (and its coloring) instead of returning \true.

We argue the correctness of the algorithm using induction on $a$.
The base case occurs when $a = 1$.
It is easy to see that in this case the algorithm correctly concludes the value of $D_{(a, q)}[E(G)]$.
Assume that the algorithm is correct for all values of $a$ that are strictly less than $p'$, for some $2 \leq p' < p$. The algorithm returns \yes\ for input $(G, l, p)$ only if Algorithm~\ref{alog:faster-randomized-algo} has concluded $D_{(p, q)}[E(G)] = $ \true.
In this case, there exists a partition $L \uplus R$ of $E(G)$ such that ${D^\star}_{(\floor{p/2},q)}[L]$ and ${D^\star}_{(\ceil{p/2},q)}[R]$ are set to \true.
By induction hypothesis, there exists $\floor{p/2}, \ceil{p/2}$  many edge-disjoint matchings, each containing $q$ edges, in $G[L]$ and $G[R]$, respectively.
This implies there are $p$ edge-disjoint matchings each containing $q$ edges.
By Observation~\ref{obs:color-to-matching},  $(G, l, p)$ is a \yes\ instance.

It remains to argue that given a \yes\ instance $(G, l, p)$, the algorithm returns \yes\ with constant probability. 
Let $\mathcal{E}(a)$ denote the event that ${D}_{(a,q)}[X]=\true$ implies ${D^\star}_{(a,q)}[X]=\true$.
Notice that $\mathcal{E}(p)$ is exactly the event where our algorithm succeeds i.e. correctly determines ${D}_{(p,q)}[E(G)]$. 
We present a lower bound on $Pr(\mathcal{E}(a))$ using  following a recurrence equation. 
We say the algorithm correctly partitions the solution edges if $L$ and $R$ contain the solution edges colored with first $\floor{a/2}$ colors and remaining  $[\ceil{a/2}]$ colors, respectively. 
The probability of success for the event depends on the following two independent events -- $(i)$ the algorithm correctly partitions in at least one of the $2^{aq}\log{(4l)}$ rounds, and $(ii)$ the values ${D^\star}_{(\floor{a/2},q)}[L]$ and ${D^\star}_{(\ceil{a/2},q}[R]$ are computed correctly. 
The probability of a partition $(L,R)$ failing to divide the solution edges ($aq$ many)  correctly in any of the rounds can be upper bounded following expression:
$$\Big(1- \frac{1}{2^{aq}}\Big)^{2^{aq} \cdot \log{(4l)}} \leq \frac{1}{2(l-1)}$$
Hence, we have the following recurrence equation:
$$Pr(\mathcal{E}(a)) \geq \Big(1-\frac{1}{2(l-1)}\Big) \cdot Pr(\mathcal{E}(\floor{a/2})) \cdot Pr(\mathcal{E}(\ceil{a/2}))$$
with $Pr(\mathcal{E}(a))=1$, when $a=1$.
The base case of the recurrence equation follows directly from the algorithm. 
The above recurrence implies $Pr(\mathcal{E}(p))\geq 1/2$ i.e. given a \yes\ instance, the algorithm returns \yes\ with probability at least $1/2$.

The runtime of the algorithm is given by the following set of recurrence equations: $T(a)=2^{aq} \log(4l) \cdot (T(\floor{a/2})+T(\ceil{a/2}))$, where $T(1)=|V(G)|^{\mathcal{O}(1)}$.
This recurrence equation solve to $ T(p)\leq 4^{l+o(l+p)}|V(G)|^{\mathcal{O}(1)}$ which gives us the running time of our algorithm.
\qed
\end{proof}

We note that the algorithm mentioned in Lemma~\ref{lemma:fpt-edges-colors-random-1} can be de-randomized using $(E(G), l)$-perfect hash families \cite{naor1995splitters}.

\subsection{A Randomized \FPT\ Algorithm}
\label{sub-sec:fpt-edges-rainbow}

In this subsection, we present a randomized \FPT\ algorithm running in time $2^l \cdot |V(G)|^{\calO(1)}$ by reducing a given instance of \maxedgecoloring\ to an instance of \textsc{Rainbow Matching}.
In \textsc{Rainbow Matching} problem, the input is an edge-labeled graph $G'$ and a positive integer $k$ and the objective is to determine whether there exists a matching of size at least $k$ such that all the edges in the matching have distinct labels.
Such matching is called as \emph{rainbow matching}.
We use the following known result.
\begin{proposition}[Theorem~$2$ in \cite{gupta2019parameterized}] \label{prop:rainbow-coloring} There exists a randomized algorithm that, given a \textsc{Rainbow Matching} instance $(G', k)$, in time $2^k \cdot |V(G')|^{\calO(1)}$ either reports a failure or finds a rainbow matching.
Moreover, if the algorithm is given a \yes\ instance, it returns a rainbow matching with constant probability. 
\end{proposition}
We use `colors' for instances of \maxedgecoloring\ and `labels' for instances of \textsc{Rainbow Matching}. 

\vspace{0.3cm}
\noindent \textbf{Reduction :} Given an instance $(G, l, p)$ of \maxedgecoloring, the reduction algorithm returns an instance $(G', k)$ of \textsc{Rainbow Matching}.
To construct graph $G'$, the algorithm creates $p$ identical copies of $G$.
Formally, for every vertex $u$ in $V(G)$, it adds $p$ vertices $u_i$ for $i \in [p]$ in $V(G')$.
For every edge $uv$, it adds all the edges $u_iv_i$ for $i \in [p]$ in $E(G')$.
The algorithm arbitrary construct a one-to-one function $\psi' : E(G) \rightarrow \{1, 2, \dots, |E(G)|\}$ on edges in $G$.
It constructs an edge-labelling function $\psi$ for edges in $G'$ in the following way : for $i \in [p]$, assign $\psi(u_iv_i) = \psi'(uv)$.
Algorithm assigns $k = l$ and returns $(G', k)$.

\begin{lemma} \label{lemma:reduction-maxedgecolor-rm} Let $(G', k)$ be the instance returned by the reduction algorithm when input is $(G, l, p)$.
Then, $(G, l, p)$ is a \yes\ instance of \maxedgecoloring\ if and only if $(G', k)$ is a \yes\ instance of \textsc{Rainbow Matching}.
\end{lemma}
\begin{proof}
$(\Rightarrow)$ 
By Observation~\ref{obs:color-to-matching}, there are $p$ many edge disjoint matchings $M_1, M_2,$ $\dots, M_p$ in $G$ such that $|M_1 \cup M_2 \cup \cdots \cup M_p| \ge l$.
We construct a rainbow matching $M'$ in $G'$ in the following way:
For $i \in [p]$, if edge $uv \in E(G)$ is in $M_i$ then add $u_iv_i$ to $M'$.
By construction, $M'$ has at least $k = l$ edges.
Since $M_i$ is a matching, there is at most one edge in $M_i$ which is incident on any vertex in $V(G)$.
Hence, an edge $u_iv_i$ is added to $M'$ then no other edge incident on $u_i$ or $v_i$ is added to $M'$.
This implies $M'$ is matching in $G'$.
We now argue that all edges in $M'$ have distinct labels.
Note that the only edges in $G'$ which has same labels are  $u_iv_i$ and $u_jv_j$ for some $uv \in E(G)$ and $i, j \in [p]$.
Since matchings $M_1, M_2, \dots, M_p$ are edge disjoint, if an edge $uv$ is present in $M_i$ then it is not present in $M_j$ for any $j \in [p] \setminus \{i\}$.
Hence, all edges in $M'$ have distinct labels.
This implies $(G', k)$ is a \yes\ instance.

$(\Leftarrow)$
Let $M'$ be a matching in $G'$ such that $|M'| \ge  k$ and every edge in $M'$ has distinct label.
By construction, every edge in $E(G')$, and hence in $M'$, is of the form $u_iv_i$ for some $i \in [p]$ and $uv \in E(G)$.
We construct $p$ matchings $M_1, M_2, \dots, M_p$ in $G$ in the following way:
For $i \in [p]$, if edge $u_iv_i$ is in $M'$ then add $uv$ to $M_i$.
Since $M'$ is a matching, if edges $u_iv_i$ are in $M'$ then no other edge incident on $u_i$ or $v_i$ is in $M'$.
Hence, for every $i \in [p]$, set $M_i$ is a matching in $G$.
We now argue that these constructed matchings are edge disjoints.
Assume, for the sake of  a contradiction, that for some $i, j \in [p]$, matchings $M_i$ and $M_j$ intersect.
Let $uv$ be the edge in $M_i \cap M_j$.
The only reason edge $uv$ is added to $M_i$ and to $M_j$ is because edges $u_iv_i, u_jv_j$ are present in $M'$.
By construction, edges $u_iv_i, u_jv_j$ have same label.
This contradicts the fact that edges in $M'$ have distinct edges.
Hence our assumption is wrong and the matchings in $G$ are pairwise disjoint.
This fact, along with the construction, implies that $|M_1 \cup M_2 \cup \cdots \cup M_p| \ge l = k$. 
By Observation~\ref{obs:color-to-matching}, $(G, l, p)$ is a \yes\ instance.
\qed
\end{proof}

Proposition~\ref{prop:rainbow-coloring} and Lemma~\ref{lemma:reduction-maxedgecolor-rm} implies that there exists a randomized algorithm that given $(G, l, p)$ either finds a subgraph $H$ of $G$ and its $p$-edge coloring such that $|E(H)| \ge l$ or correctly concludes that no such subgraph exists in time $\calO^*(2^l)$.
Moreover, if such a subgraph exists in $G$, then the algorithm returns it with constant probability.
\section{Kernelization Algorithm}
\label{sec:kernel}

In this section, we prove that \maxedgecoloring\ admits a polynomial kernel when parameterized by the number of colors and $|X|$ where $X$ is a minimum sized \text{deg}-$1$-modulator.
As discussed in Section~\ref{sec:prelims}, such result implies that the problem admits a polynomial kernel when parameterized by the number of colors together with one of the following parameters:
$(1)$ the number of edges, $l$, in the desired subgraph, 
$(2)$ the vertex cover number of the input graph $\vc(G)$, and 
$(3)$ the above guarantee parameter $(l - \mm(G))$.
Our kernelization algorithm is based on the expansion lemma.

Consider an instance $(G, p, l)$ of \maxedgecoloring.
We assume that we are given a \text{deg}-$1$-modulator $X$ of $G$ (see Definition~\ref{def:deg-1-mod}).
We justify this assumption later and argue that one can find a \text{deg}-$1$-modulator which is close to a minimum sized \text{deg}-$1$-modulator in polynomial time.
We start with the following simple reduction rule.
\begin{reduction-rule}\label{rr:delete-isolated-components} If there exists a connected component $C$ of $G - X$ such that no vertex of $C$ is adjacent to a vertex in $X$, then delete all the vertices in $C$ and reduce $l$ by $|E(C)|$, i.e. return the instance $(G - V(C), l - |E(C)|, p)$.  
\end{reduction-rule}
\begin{lemma}
  \label{lemma:delete-isolated-components}
  Reduction Rule~\ref{rr:delete-isolated-components} is safe and given set $X$, it can be applied in polynomial time.
\end{lemma}
Let $(G, l, p)$ be the instance obtained by exhaustively applying Reduction Rule~\ref{rr:delete-isolated-components}. This implies that every connected component of $G - X$ is adjacent to $X$.
Let $\calC$ be the set of connected components of $G - X$. 
We construct an auxiliary bipartite graph $B$, with vertex bipartition $X$ and  $\calC$ (each $C \in \calC$ corresponds to a vertex, say $b_C$ of $B$).  
There exists edge $xb_C$ in $B$ for $x \in X$ and $b_C \in \calC$ if and only $x$ is adjacent to at least one vertex in $C$ in $G$.
For $\calC' \subseteq \calC$ of connected components, $V(\calC') \subseteq V(G)$ denotes the vertices in connected components in $\calC'$ and $E(\calC') \subseteq E(G)$ denotes the edges that have both endpoints in $V(\calC')$. Since every connected component in $\calC$ is adjacent to $X$, there are no isolated vertices in $B$. 
We can thus apply the following rule which is based on the Expansion Lemma.

\begin{reduction-rule} \label{rr:expansion-lemma-component} 
  If $|\calC| \ge p|X|$ then apply
  Lemma~\ref{lem:expansion-lemma} to find $X' \subseteq X$ and $\calC' \subseteq \calC$ such that $(1)$ there exits a $p$-expansion from $X'$ to $\calC'$; and $(2)$ no vertex in $\calC'$ has a neighbour outside $X'$.
Delete all the vertices in  $X' \cup V(\calC')$ from $G$ and reduce $l$ by
$p|X'| + |E(\calC')|$, i.e. return $(G - (X' \cup V(\calC')), l - p|X'| - |E(\calC')|, p)$. 
\end{reduction-rule}

\begin{lemma} \label{lemma:expansion-lemma-component}
Reduction Rule~\ref{rr:expansion-lemma-component} is safe and given set $X$, it can be applied in polynomial time.
\end{lemma}
\begin{proof} 
Let $M'$ be the edges in $p$-expansion lemma from $X'$ to $\calC'$ for the bipartite graph $B$. 
We construct a set $M \subseteq E(G)$, corresponding to edges in as follows $M'$. 
If there is an edge $xb_C$ in $M'$ then pick an edge whose one endpoint is $x$ and another endpoint is in $C$.
If there are multiple such edges then arbitrarily pick one of them.
Consider a subgraph $H_1$ of $G$ such that $V(H_1) = X' \cup V(\calC')$ and $E(H_1) = M \cup E(\calC')$. 
A $p$-star graph is a tree on $p + 1$ vertices such that there exists a vertex that is adjacent to all other vertices. 
Notice that each connected component of $H_1$ is a tree. 
Since every connected component in $\calC'$ has at most one edge, each tree in $H_1$ can be obtained from a $p$-star by adding (at most one) new vertex and making it adjacent with one of its leaves. It is easy to see that $H_1$ is a $p$-edge colorable graph and every vertex in $X'$ is of degree $p$ in $H_1$. 

Suppose there exists exists a subgraph $H'$ of $G - (X' \cup V(\calC')$ such that $H'$ is $p$-edge colorable and has at least $l - p|X| - |E(\calC')|$ edges.
Then, $H' \cup H_1$ is a subgraph of $G$ which is $p$-edge colorable and has at least $l $  edges. 
Here, $H' \cup H_1$ denote the graph with vertex set $(V(H) \cup V(H_1)$ and edges $E(H) \cup E(H_1))$.

Suppose there exists a subgraph $H$ of $G$ which is $p$-edge colorable and has at least $l$ edges.
By Proposition~\ref{thm:vizing}, the maximum degree of a vertex in $H$ is $p$.
Let $H^{\circ}$ be the graph obtained from $H$ by deleting all vertices in $X' \cup V(\calC')$. 
Since $H^{\circ}$ is a subgraph of $H$, it is $p$-edge colorable. Note that $H^{\circ}$ is also a subgraph of $G - (X' \cup V(\calC'))$.
To complete the proof, we need to argue that $H^{\circ}$ has at least $|E(H)| - p|X'| + |E(\calC')|$ edges. 
Since every vertex in $H$ has degree at most $p$, there are at most $p|X'|$ many edges across $X, V(\calC')$ i.e. edges with one vertex in $X$ and another in $V(\calC')$.
Moreover, there are at most $|E(\calC')|$ edges in $H$ whose both endpoints are in $V(\calC')$.
Since $V(\calC')$ are adjacent with vertices only in $X'$, there are no other edges incident on $V(\calC')$. 
This implies number of edges in $H^{\circ}$ is at least $|E(H)| - p|X'| + |E(\calC')|$ which concludes the proof.\qed
\end{proof}

In the following lemma, we argue that \maxedgecoloring\ admits a polynomial kernel when parameterized by size of the given \text{deg}-$1$-modulator.

\begin{lemma}\label{lemma:kernel-deg-1-modulator}
Consider an instance $(G, l, p)$ and let $X$ be a \text{deg}-$1$-modulator of $G$.
Then, \maxedgecoloring\ admits a kernel with $\calO(|X|p)$ vertices. 
\end{lemma}
\begin{proof}
The algorithm then applies Reduction Rule~\ref{rr:delete-isolated-components} and \ref{rr:expansion-lemma-component} exhaustively.
It returns the reduced instance as a kernel.
We now argue the correctness and the size bound on the reduced instance.
Let $(G', l', p)$ be the reduced instance obtained by the algorithm after exhaustive application of reduction rules on input instance $(G, l, p)$.
By Lemma~\ref{lemma:delete-isolated-components} and \ref{lemma:expansion-lemma-component}, $(G, l, p)$ is a \yes\ instance if and only if $(G', l', p')$ is a \yes\ instance.
Moreover,  since reduction rules are not applicable, the number of vertices in $G'$ is at most $(p + 1) |X|$.
\qed\end{proof}

\begin{proof} (of Theorem~\ref{thm:kernel})
For an instance $(G, l, p)$ of \maxedgecoloring\ the kernelization algorithm first uses Observation~\ref{obs:para-hierarchy} to conclude that either $(G, l, p)$ is a \yes\ instance or $\vc(G) \preceq l$ and $|X_{opt}| \preceq (l - \mm(G))$, where $X_{opt}$ is a minimum sized \text{deg}-$1$-modulator of $G$.
In the first case, it returns a vacuously true instance of constant size.
If it can not conclude that given instance is a \yes\ instance then algorithm computes a \text{deg}-$1$-modulator, say $X$, of $G$ using the simple $3$-approximation algorithm: there exists a vertex $u$ which is adjacent with two different vertices, say $v_1, v_2$ then algorithm adds $u, v_1, v_2$ to the solution. 
It keeps repeating this step until every vertex is of degree at most one. 
The algorithm uses the kernelization algorithm mentioned in Lemma~\ref{lemma:kernel-deg-1-modulator} to compute a kernel of size $\calO(p |X|)$.

The correctness of the algorithm follows from the correctness of Lemma~\ref{lemma:kernel-deg-1-modulator}.
As $X$ is obtained by using a $3$-factor approximation algorithm,  $|X| \le 3 |X_{opt}|$ and hence $|X| \preceq |X_{opt}|$.
Since the algorithm was not able to conclude that $(G, l, p)$ is a \yes\ instance, by Observation~\ref{obs:para-hierarchy}, we have $|X_{opt}| \preceq \vc(G) \preceq l$ and $|X_{opt}| \preceq l - \mm(G)$.
This implies the number of vertices in the reduced instance is at most $\calO(kp)$ where $k$ is one of the parameters in the statement of the theorem.
By Observation~\ref{obs:para-hierarchy}, Lemma~\ref{lemma:delete-isolated-components} and Lemma~\ref{lemma:expansion-lemma-component}, and the fact that every application of reduction rules reduces the number of vertices in the input graph, the algorithm terminates in polynomial time.
\qed\end{proof}

\section{Lower Bounds on the Size of  Kernels}
\label{sec:kernel-lower-bound}
The objective of this section is to prove Theorem~\ref{thm:kernel-lower-bound-solution-size}. Due to Observation~\ref{obs:para-hierarchy}, it is sufficient to prove such result for the number of edges in the desired graph.
In other words, we prove that for any $\epsilon > 0$ and computable function $f$, \maxedgecoloring\ does not admit a polynomial compression of size $\calO(l^{1 - \epsilon} \cdot f(p))$ unless $\NP \subseteq \coNP/poly$. We obtain the above result by giving an appropriate reduction from \redbluedsfull\ to \maxedgecoloring. 

\redbluedsfull\ (\redblueds, for short) takes as input a bipartite graph $G$, with vertex bi-partitions as $R,B$ and an integer $k$, and the objective is to decide if there is $R' \subseteq R$ of size at most $k$ such that for each $b \in B$, $R' \cap N(b) \neq \emptyset$.\footnote{The sets $R$ and $B$ are referred as red and blue sets, respectively.}
Without loss of generality, we can assume that there are no isolated vertices in the input graph.
The problem \dset\ takes as an input a graph $G$ and an integer $k$, and the goal is to decide whether there exists $X \subseteq V(G)$ of size at most $k$, such that for each $v \in V(G)$, $X \cap N[v] \neq \emptyset$. Jansen and Pieterse proved that \dset\ does not admit a compression of bit size $\OO(n^{2-\epsilon})$, for any $\epsilon >0$ unless $\NP \subseteq \coNPpoly$, where $n$ is the number of vertices in the input graph~\cite{DBLP:conf/iwpec/JansenP15}. This result directly implies the following (see, for instance~\cite{agrawal2017paths}, for a formal statement). 

\begin{proposition}\label{prop:rbds-no-subquadratic-compression}
\redbluedsfull\ does not admit a compression of bit size $\OO(n^{2-\epsilon})$, for any $\epsilon >0$, unless ${\mbox{\sf \NP}} \subseteq {\mbox{\sf \coNPpoly}}$. Here, $n$ is the number of vertices in the input graph.
\end{proposition}

\begin{figure}[t]
  \begin{center}
    \includegraphics[scale=0.5]{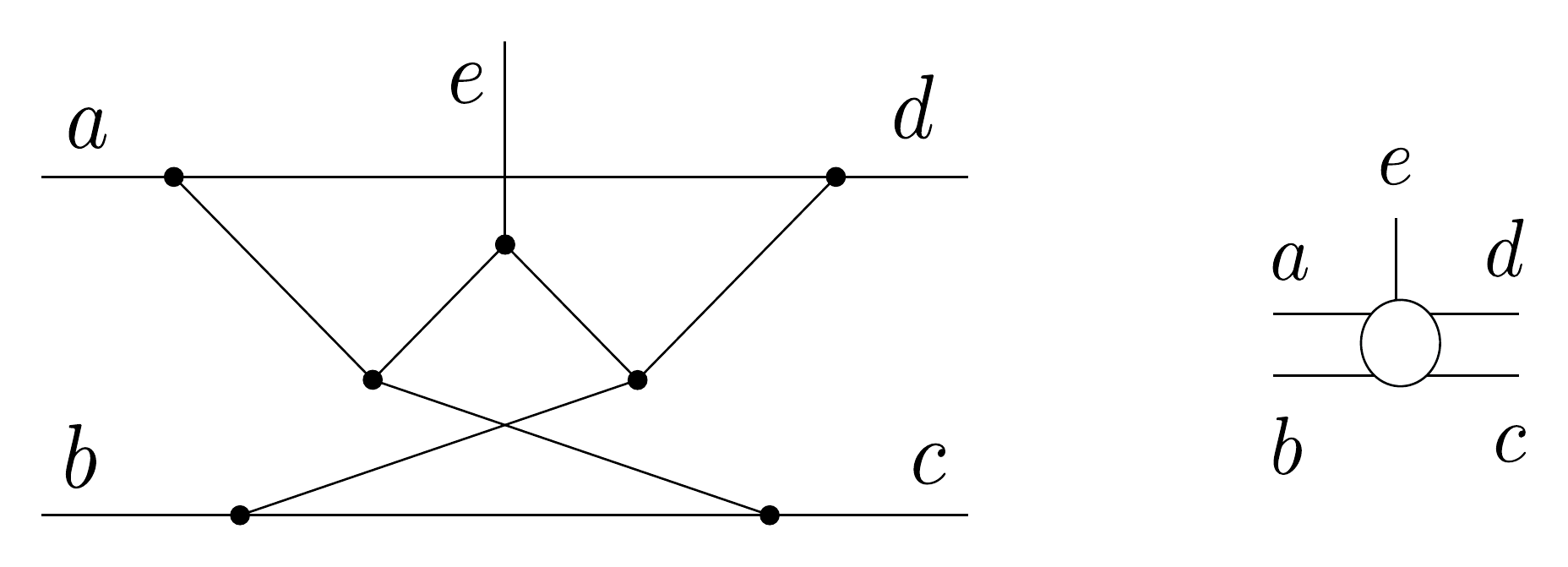}
  \end{center}
  \caption{Inverting components and its symbolic representation used in the reduction from \redbluedsfull\ to \maxedgecoloring. \label{fig:module-red-vertex}}
\end{figure}

Holyer showed that it is \NP-hard\ to distinguish whether the given cubic graph admits a $3$-edge coloring, or any edge coloring of it requires $4$ colors)~\cite{holyer1981computational} (also see Proposition~\ref{thm:vizing}).
Laven and Galil generalized this result to prove that for any fixed $p$, the problem of deciding whether the edge chromatic number of a regular graph of degree $p$ is $p$ or $p + 1$ \cite{leven1983np}.
We start with an inverting component presented in \cite{holyer1981computational} (see Figure~\ref{fig:module-red-vertex}).
We call this graph as a \emph{module}. 
Note that $a, b, c, d,$ and $e$ are labelings of the corresponding edges, and the other endpoints of these edges are not shown in the figure.
We state following two useful properties of modules.
\begin{claim1}[{\cite[Lemma~$1$]{leven1983np}}] \label{claim:module-color-prop}
    For any $3$-edge coloring $\phi$ of a module, either $\phi(a) = \phi(b)$ or $\phi(c) = \phi(d)$. Moreover, if $\phi(a) = \phi(b)$ then $\phi(c), \phi(d), \phi(e)$ are all different; else $\phi(c) = \phi(d)$, and then $\phi(a), \phi(b), \phi(e)$ are all different. 
\end{claim1}
\begin{claim1}[{\cite[Lemma~$1$]{leven1983np}}]
  \label{claim:module-color-extension}
    Consider a partial $3$-coloring $\phi$ of edges in a module which satisfy either of two conditions: $(1)$ $\phi(a) = \phi(b)$ and $\phi(c), \phi(d), \phi(e)$ are all different; $(2)$ $\phi(c) = \phi(d)$ and $\phi(a), \phi(b), \phi(e)$ are all different. Then, $\phi$ can be extended to a $3$-edge coloring of the module.    
  \end{claim1}
We next present a polynomial time reduction from \redblueds\ to \maxedgecoloring. Consider an instance $(G, R, B, k)$ of \redblueds. We construct an instance $(G', l, p)$ of \maxedgecoloring, as follows. 

\vspace{0.5cm}
\noindent \textbf{Reduction :} Initialize $V(G') = V(G) = R \cup B$ and $E(G') =\emptyset$. For every vertex $r \in R$, we construct a gadget using $2 \cdot (\text{deg}_G(r) + 1)$ modules as shown in Figure~\ref{fig:both-gadget}. This gadget has $(\text{deg}_G(r) + 1)$ many pairs of edges which acts as outputs. Arbitrarily fix a pair edges in outputs and make $r$ an endpoint of both these edges.
We call this gadget a \emph{red gadget} corresponding to $r$. 
For every blue vertex $b \in B$, construct a cycle of length $(2 \cdot \text{deg}_G(r) + 1)$. Arbitrarily fix a vertex on this cycle and make it adjacent with $b$. Add $\text{deg}_G(b)$ many modules to this cycle such that edges on the cycles are endpoints of pairs of edges which are outputs of these modules (see Figure~\ref{fig:both-gadget}). Add these modules in such that after addition, the degree of every vertex on cycle is three. We call this gadget a \emph{blue-gadget} corresponding to $b$.
For each edge $rb$ in $G$, identify a pair of edges in outputs of red-gadget corresponding to $r$ with a pair of edges in inputs of blue-gadget corresponding to $b$.  
In other words, the other endpoints of edges in red-gadget corresponding to $r$ is in blue-gadget corresponding to $b$ and vice-versa (see the edges $e_1, e_2$ in Figure~\ref{fig:both-gadget}).
This completes the construction of graph $G'$. Assign $p = 3$, $l = |E(G')| - k$ and return $(G', l, p)$ as an instance of \maxedgecoloring. 

\begin{figure}[t]
  \begin{center}
    \includegraphics[scale=0.5]{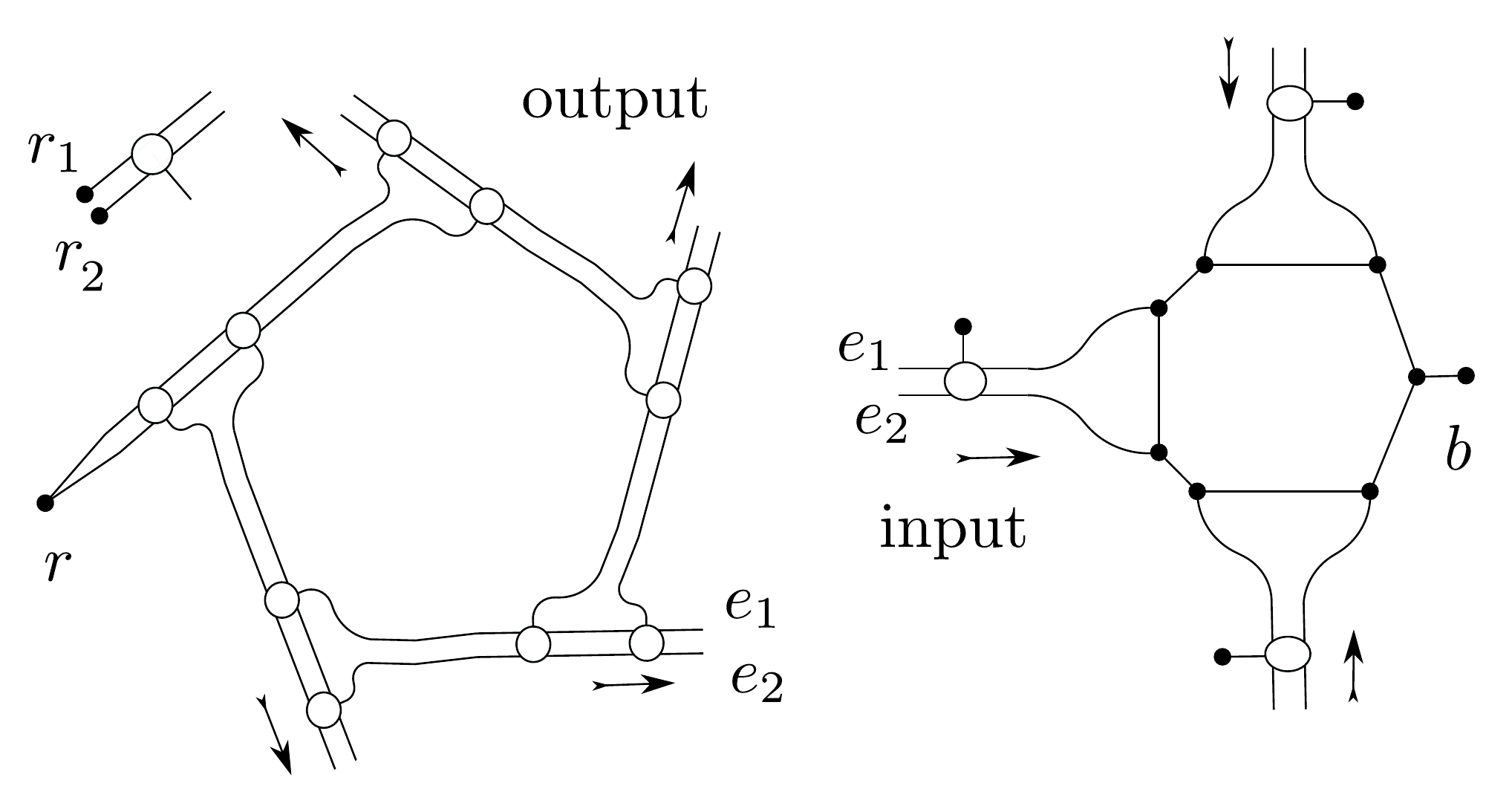}
  \end{center}
  \caption{(Left Figure) A red-gadget made for a vertex of degree four. The gadget is made from ten modules and have four output pairs of edges. More generally, it can be made from $2(d + 1)$ modules and has $d$ output pairs of edges. In the modified red-gadget, vertex $r$ is replaced by two vertices $r_1, r_2$ as shown in left-top corner. (Right Figure) A blue-gadget made for a vertex of degree three. More generally, it can be made from $d$ modules and has $d$ input pairs of edges.
 \label{fig:both-gadget}}
\end{figure} 

Next we argue that thenumber of edges in $G'$ is at most constant times the number of edges in $G$.
\begin{lemma}\label{lemma:edges-bound} 
We have $|E(G')| \le c \cdot |E(G)|$, where $c$ is a (fixed) constant. 
\end{lemma}
\begin{proof} Every module contains seven vertices each of which has degree three. Hence there are at most $21$ edges which are incident on vertices in a module. Since red-gadget corresponding to $r$ uses $2 \cdot (\text{deg}_G(r) + 1)$ many gadgets, the number of edges incident on this red-gadget is at most $42 \cdot (\text{deg}_G(r) + 1)$. Hence, the total number of edges incident on red-gadgets is at most $\sum_{r \in R} 42 \cdot (\text{deg}_G(r) + 1) \le 42 \cdot |E(G)| + |R|$.   Similarly, blue-gadget corresponding to $b$ uses $\text{deg}_G(b)$ modules, a cycle with edges $2 \cdot \text{deg}_G(r) + 1$, and an extra edge. Hence the total number of edges incident on vertices in blue-gadgets is at most $\sum_{b \in B} (21 \cdot\ \text{deg}_G(b) + 2 \cdot \text{deg}_G(b) + 2 ) \le 23|E(G)| + 2|B|$. Since there are no isolated vertices in $G(R, B)$, we can conclude that total number of edges in $G'$ is at most $67 |E(G)|$. This concludes the proof of the lemma.
\qed\end{proof}

To simplify our arguments in the proof of correctness of the reduction, we construct an auxiliary graph. 
We define the following process which takes the graph $G'$ and a subset $R'$ of $R$ and returns another graph $G''$ such that $|V(G'')| = |V(G')| + |R'|$ and $|E(G'')| = |E(G')|$. 

\vspace{0.5cm}
\noindent \textbf{Modification of $G'$ at $R'$:} For every vertex $r$ in $R'$ do the following process. Add two vertices $r_1, r_2$ to $G'$ and delete $r$ from $G$. Let $x_1,  x_2$ be the two vertices in red-gadget which were adjacent with $r$ in $G'$. Add edges $r_1x_1$ and $r_2x_2$ (see Figure~\ref{fig:both-gadget}). Let $R'_i = \{r_i |\ r \in R'\}$ for $i \in \{1, 2\}$. In other words, $G''$ is obtained from $G$ by deleting all vertices in $R'$ and adding vertices in $R'_1 \cup R'_2$, and making them adjacent with the corresponding neighbors of vertices in $R$. If red-gadget is modified at $r$ in $R'$ then we call it the \emph{modified red-gadget} at $r$. 
\vspace{0.5cm}

In any $3$-edge-coloring of a red-gadget, edges incident on vertices in $R$ are of different colors.
With this observation, {\cite[Lemma~$2$]{leven1983np}} implies following two properties of red-gadgets.
\begin{claim1} 
  \label{claim:red-gadget-output}
  In any $3$-edge coloring of a red-gadget, every pair of output edges are colored with different colors.
\end{claim1}
We modify the red-gadgets to ensure that every pair of output edges can be colored with the same color. Following claim is also implied by {\cite[Lemma~$2$]{leven1983np}}.
\begin{claim1}
  \label{claim:modified-red-gadget-output}
  There exists a $3$-edge coloring of modified red-gadget such that every pair of output edges are colored with same colors.
\end{claim1}
We mention following property of blue-gadget before mentioning a relation between $G$ and $G''$.
\begin{claim1} [{\cite[Lemma~$3$]{leven1983np}}]
\label{claim:blue-gadget-input}
In any $3$-coloring of a blue-gadget at least one pair of input must be colored with the some color.
Moreover, any $3$-coloring of the input edges which satisfied the previous condition can be completed to a $3$-coloring of the gadget.
\end{claim1}
\begin{lemma} \label{lemma:aux-graph-3-colorable}
 Let $(G', l, p)$ be the instance returned by the reduction algorithm when input is $(G, R, B, k)$. For a subset $R' \subseteq R$, let $G''$ be the graph obtained by modifying $G'$ at $R'$. Then, $R'$ is adjacent with all vertices in $B$ if and only if $G''$ is $3$-edge colorable.
\end{lemma}
\begin{proof}
    $(\Rightarrow)$ We construct a $3$-edge coloring of graph $G''$. We first color all edges incident on red-gadgets and modified red-gadgets followed by edges in blue-gadgets.
Note that the sets of edges incident on red-gadgets and modified red-gadgets do not intersect with each other.
By Claim~\ref{claim:module-color-prop} and \ref{claim:module-color-extension}, there exists a coloring of red-gadgets and modified red-gadgets.
Moreover, by Claim~\ref{claim:modified-red-gadget-output}, it is safe to consider a $3$-edge coloring of modified red-gadget in which every pair of output edges are colored with the same colors.
Since $R'$ is adjacent with every vertex in $B$, every blue-gadget has at least one pair of input edges which are colored with some color. By Claim~\ref{claim:blue-gadget-input}, this coloring can be extended to other edges incident on blue-gadgets. This completes a $3$-edge coloring of $G''$. 

  $(\Leftarrow)$ Consider a $3$-edge coloring of graph $G''$.  By Claim~\ref{claim:blue-gadget-input}, for any blue-gadget, at least one pair of input must be colored with some color. By Claim~\ref{claim:red-gadget-output}, for a red-gadget, any pair of output edges are colored with different colors. Hence, for any blue-gadget, there is at least one input pair of edges that are connected to an output pair of edges of a modified red-gadget. By construction, this implies that for any vertex in $B$, there exists an edge with some vertex in $R'$. This implies that vertices in $R'$ are adjacent with every vertex in $B$.
\qed\end{proof}

In the following lemma, we argue that the reduction is safe.

\begin{lemma}\label{lemma:rbds-max-color}  Let $(G', l, p)$ be the instance returned by the reduction algorithm when $(G, R, B, k)$ is given as input. Then, $(G, R, B, k)$ is a \yes\ instance of \redblueds\ if and only if $(G', l, p)$ is a \yes\ instance of \maxedgecoloring.
\end{lemma}
\begin{proof} The problem of determining whether $(G', l, p)$ is a \yes\ instance of \maxedgecoloring\ is equivalent to determining whether one can delete at most $(|E(G')| - l)$ many edges in $G'$ such that the resultant graph is $p$-edge colorable. We work with this formulation of the problem. Since $(G', l, p)$ is an instance returned by the reduction algorithm, $|E(G')| - l \le k$ and $p = 3$.

  Let $R'$ be a subset of $R$. Define $E_{R'}$ as the set of edges in $E(G')$ formed by selecting exactly one edge incident every vertex in $R'$. By construction, $|E_{R'}| = |R'|$.  
Let $G''$ be the graph obtained from modification of $G'$ at $R'$ as specified above. Note that the graph obtained from $G'$ by deleting all edges in $E_{R'}$ is isomorphic to the graph obtained from $G''$ by deleting all vertices in $R'_1$. Here, $R_1'$ is a set defined in modification process.

In graph $G''$, every vertex in $R'_1$ is pendant vertex and is adjacent with a vertex of degree two. Hence, any $3$-edge coloring of $G'' - R'_1$ can be trivially extended to a $3$-edge coloring of $G''$. Also, as $G'' - R'_1$ is a subgraph of $G''$, any $3$-edge coloring of $G''$ is also a $3$-edge coloring of $G'' - R'_1$. Hence, $G'' - R'_1$ is $3$-edge colorable if and only if $G''$ is $3$-edge colorable. Lemma~\ref{lemma:aux-graph-3-colorable} implies that $G' - E_{R'}$ is $3$-edge colorable if and only if $R'$ is adjacent with all vertices in $B$.
Since graphs $G'' - R'_1$ and $G' - E_{R'}$ are isomorphic to each others, we get $G' - E_{R'}$ is $3$-edge colorable if and only if $R'$ is adjacent with all vertices in $B$. Since $|E_{R'}| = |R'|$, this concludes the proof of the lemma.
\qed\end{proof}

We are now in a position to present a proof for Theorem~\ref{thm:kernel-lower-bound-solution-size}. 

\begin{proof}(for Theorem~\ref{thm:kernel-lower-bound-solution-size})
For the sake of contradiction, assume that there exists an $\epsilon > 0$ and some computable function $f$ such that \maxedgecoloring\ admits a compression of size $\calO(l^{1-\epsilon} \cdot f(p))$. This implies there is an algorithm $\calA$ which takes an instance $(G', l, p)$ of \maxedgecoloring\ and in polynomial time returns an equivalent instance for some problem which needs $\calO(l^{1-\epsilon} \cdot f(p))$ bits to encode. 
  
Let $(G,R,B,k)$ be an instance of \redblueds, where $G$ is a graph on $n$ vertices. Using the reduction described, we create an instance $(G', l, p)$ of \maxedgecoloring.
It is easy to see from the description of the reduction that this instance can be created in time polynomial in the size of the given instance of \redblueds. 
By Lemma~\ref{lemma:rbds-max-color}, instances $(G, R, B, k)$ and $(G', l, p)$ are equivalent. 
On instance $(G', l, p)$, we run the algorithm $\calA$ mentioned in previous paragraph to obtain an equivalent instance of size $\calO(l^{1-\epsilon} \cdot p)$.
Note that this instance is equivalent to the given instance of \redblueds. 
Since $p = 3$ and $l \le |E(G')| \in \calO(|E(G)|) \in \calO(n^{2})$ (by Lemma~\ref{lemma:edges-bound}), this instance is of size $\calO(n^{2 - 2\epsilon})$.
This implies there exists an algorithm which in polynomial time returns an equivalent instance of \redblueds\ of size $\calO(n^{2 - 2\epsilon})$. This is a contradiction to Proposition~\ref{prop:rbds-no-subquadratic-compression}. 
Hence our assumption was wrong and \textsc{Maximum Edge-Colorable Subgraph} does not admit a compression of size $\calO(l^{1-\epsilon} \cdot f(p))$. 
The proof of the theorem follows from Observation~\ref{obs:para-hierarchy}. \qed
\end{proof}
\section{Conclusion}
\label{sec:conclusion}

In this article, we studied the \maxedgecoloring\ problem from the lense of Parameterized Complexity.
We showed that the problem admits a kernel with $\calO(k \cdot p)$ vertices where $p$ is the number of colors and $k$ is one of the following: $(a)$ the number of edges, $l$, in a desired subgraph, $(b)$ the vertex cover number of input graph, and $(c)$ the difference between $l$ and the size of a maximum matching in the graph. Furthermore, we complimented the above result by establishing that \maxedgecoloring\ does not admit a polynomial compress of size $\calO(k^{1 - \epsilon} \cdot f(p))$ for any $\epsilon > 0$ and any computable function $f$, unless $\NP \subseteq \coNP/poly$.
It will be interesting to close the gap between the kernel lower bound and the size of the kernel. As a consequence of the above kernelization results, we can obtain that the problem has a polynomial kernel when parameterized by $\ell$. It will interesting to investigate whether the problem has a polynomial kernel when parameterized by the vertex cover number, or the difference between $l$ and the size of a maximum matching in the input graph. 

We also designed \FPT\ algorithms for the parameters, $\ell$ and the vertex cover number. We leave it as an open question to determine whether the problem admits an \FPT\ algorithm when parameterized by the difference between $l$ and the size of a maximum matching in the input graph.
%
%
%
\bibliographystyle{splncs04}
\bibliography{references}
\end{document}